%% file: MonadicDatalogContainmentOnTrees.tex
\documentclass[runningheads,envcountsame]{llncs}

\usepackage{graphicx}

\usepackage{hyperref}
\urldef{\mailAndreNicole}\path|{afrochaux,schweika}@informatik.uni-frankfurt.de|    
\urldef{\mailMartinGrohe}\path|grohe@informatik.rwth-aachen.de|

\input{def.tex}
\let\doendproof\endproof \renewcommand\endproof{~\hfill$\qed$\doendproof} 

\begin{document}

\mainmatter 
\title{Monadic Datalog Containment on Trees\thanks{This article is the full version of \cite{FGS}.}}
\author{Andr\'e Frochaux\inst{1} \and Martin Grohe\inst{2} \and Nicole Schweikardt\inst{1}}
\institute{
 Goethe-Universit\"at Frankfurt am Main, \mailAndreNicole
 \and RWTH Aachen University, \mailMartinGrohe}

\maketitle

\input{main.tex}

\noindent\textbf{Open Question.} It remains open to close the gap between the \EXPTIME\ lower and the \TwoEXPTIME\ upper bound for the case where the \emph{descendant}-axis is involved.
\smallskip

\noindent\textbf{Acknowledgment.} The first author would like to thank Mariano Zelke for countless inspiring discussions and helpful hints
on and off the topic.

\enlargethispage{1cm}

\bibliographystyle{abbrv}
\bibliography{MonadicDatalogContainmentOnTrees}

\clearpage

\input{appendix.tex}

\end{document}

%% file: def.tex
\usepackage{amsmath}

\usepackage{amssymb}
\usepackage{amsfonts}
\usepackage{stmaryrd}
\usepackage{enumerate}
\usepackage{multicol}
\usepackage[lined,boxed]{algorithm2e}

\usepackage{tikz}
\usetikzlibrary{arrows}

\usetikzlibrary{arrows,positioning} 
\tikzset{
    >=stealth',
    punkt/.style={

           text width=13em,
           minimum height=2em,
           text centered},
    pil/.style={
           ->,
           thick,
           shorten <=2pt,
           shorten >=2pt,}
}

\newcounter{ruleBoxNo}

\usepackage{framed, color}

\definecolor{shadecolor}{gray}{.95}
\definecolor{framecolor}{gray}{.65}

{ \rightline{\textit{Box \theruleBoxNo}} \endMakeFramed}

          \newcommand{\EXPTIME}{\textnormal{\textnormal{\textsc{Exptime}}}}
                    \newcommand{\TwoEXPTIME}{\textnormal{\textnormal{\textsc{2Exptime}}}}
          \newcommand{\T}{\ensuremath{\mathcal{T}}}

\newcommand{\AF}[1]{{#1}}

          \newcommand{\lang}{\ensuremath{\mathcal{L}}}

          \newtheorem{Theorem}[theorem]{Theorem}
          \newtheorem{Proposition}[theorem]{Proposition}
          
          \newtheorem{Fact}[theorem]{Fact}
          \newtheorem{Lemma}[theorem]{Lemma}
          
          \newtheorem{Corollary}[theorem]{Corollary}

          \newtheorem{Def}[theorem]{Definition}

          \newtheorem{Claim}{Claim} 
		  
          \DeclareMathOperator{\MSO}{\textnormal{MSO}}
          
          \DeclareMathOperator{\mDatalog}{\textnormal{mDatalog}}

          \DeclareMathOperator{\Root}{\textnormal{\textbf{root}}}

          \DeclareMathOperator{\Lc}{\textnormal{\textbf{lc}}}

          \DeclareMathOperator{\Rc}{\textnormal{\textbf{rc}}}
          \DeclareMathOperator{\Label}{\textnormal{\textbf{label}}}
          \DeclareMathOperator{\Child}{\textnormal{\textbf{child}}}

          \DeclareMathOperator{\Desc}{\textnormal{\textbf{desc}}}

          \DeclareMathOperator{\Fc}{\textnormal{\textbf{fc}}}

          \DeclareMathOperator{\Ns}{\textnormal{\textbf{ns}}}

          \DeclareMathOperator{\Ls}{\textnormal{\textbf{ls}}}
          \DeclareMathOperator{\Leaf}{\textnormal{\textbf{leaf}}}

          \DeclareMathOperator{\Hnrc}{\textnormal{\textbf{has\_no\_rc}}}
          \DeclareMathOperator{\Hnlc}{\textnormal{\textbf{has\_no\_lc}}}
          \DeclareMathOperator{\Isrc}{\textnormal{\textbf{is\_rc}}}
          \DeclareMathOperator{\Islc}{\textnormal{\textbf{is\_lc}}}

          \DeclareMathOperator{\Accept}{\textnormal{\textbf{accept}}}
          \DeclareMathOperator{\Reject}{\textnormal{\textbf{reject}}}

          \DeclareMathOperator{\idb}{\textnormal{idb}}
          \DeclareMathOperator{\edb}{\textnormal{edb}}
          \DeclareMathOperator{\Var}{\textnormal{var}}

          \newsavebox{\probbox}  
          \newlength{\problength}  
          \newenvironment{Problem}[1]{        
            \begin{center}  
            \setlength{\problength}{0.9\textwidth} 
            \begin{lrbox}{\probbox}        
            \begin{minipage}{\problength}   
               \textnormal{\textsc{#1}}     
                \begin{quote}    
                \begin{description}  
                }{     
                \end{description}         
                  \end{quote}  
            \end{minipage}\end{lrbox}\noindent\fbox{\usebox{\probbox}}   
            \end{center}    
           }                                                                            
          \newcommand{\In}{\item[\textnormal{\textit{Input:}}]}

          \newcommand{\Quest}{\item[\textnormal{\textit{Question:}}]}

          \newcommand{\Yes}{\ensuremath{\textnormal{\textbf{yes}}}} 
          \newcommand{\No}{\ensuremath{\textnormal{\textbf{no}}}} 
          \newcommand{\True}{\ensuremath{\textnormal{\textbf{true}}}} 
          \newcommand{\False}{\ensuremath{\textnormal{\textbf{false}}}}

          \renewcommand{\phi}{\varphi}

          \newenvironment{mi}{\begin{itemize}}{\end{itemize}}
          \newenvironment{me}{\begin{enumerate}[(1)]}{\end{enumerate}}
          \newenvironment{mea}{\begin{enumerate}[(a)]}{\end{enumerate}}
          \newenvironment{mei}{\begin{enumerate}[(i)]}{\end{enumerate}}

          \newenvironment{proofsketch}{\begin{proof}[sketch]}{\end{proof}}

          \newenvironment{proofof}[1]{ \bigskip \noindent\textbf{Proof
              of #1:}\\ }{\qed}

          \newcommand{\nc}{\newcommand}
          \nc{\rnc}{\renewcommand}

          \rnc{\geq}{\ensuremath{\geqslant}}
          \rnc{\leq}{\ensuremath{\leqslant}}

          \nc{\deff}{:=}

          \nc{\set}[1]{\ensuremath{\{ #1 \}}}
          \nc{\setc}[2]{\set{#1 \,:\, #2}}

          \nc{\isom}{\ensuremath{\cong}}

          \nc{\dist}{\ensuremath{\textit{dist}}}

          \nc{\ov}[1]{\ensuremath{\overline{#1}}}
          \nc{\mymod}{\ensuremath{\textup{ mod }}}

          \nc{\NN}{\ensuremath{\mathbb{N}}}
          \nc{\Z}{\ensuremath{\mathbb{Z}}} 
          \nc{\NNpos}{\ensuremath{\NN_{\mbox{\tiny $\scriptscriptstyle \geq 1$}}}}

          \nc{\und}{\ensuremath{\wedge}}
          \nc{\Und}{\ensuremath{\bigwedge}}
          \nc{\oder}{\ensuremath{\vee}}
          \nc{\Oder}{\ensuremath{\bigvee}}
          \nc{\nicht}{\ensuremath{\neg}}
          \nc{\impl}{\ensuremath{\rightarrow}}
          \nc{\gdw}{\ensuremath{\leftrightarrow}}

          \nc{\free}{\ensuremath{\textit{free}}}
          \nc{\ar}{\ensuremath{\textit{ar}}} 
          \nc{\qr}{\ensuremath{\textit{qr}}}

          \nc{\A}{\ensuremath{\mathcal{A}}}
          \nc{\B}{\ensuremath{\mathcal{B}}}
          \nc{\C}{\ensuremath{\mathcal{C}}}
          \nc{\PP}{\ensuremath{\mathcal{P}}}
          \rnc{\S}{\ensuremath{\mathcal{S}}}

          \nc{\SGK}{\ensuremath{\S_{\textit{GK}}}}

          \nc{\tauGK}{\ensuremath{\tau_{\textit{GK}}}}
          \nc{\tauGKAlph}[1]{\ensuremath{\tau_{\textit{GK},#1}}}
          \nc{\tauGKSigma}{\tauGKAlph{\Sigma}}

          \nc{\tauo}{\ensuremath{\tau_o}}
          \nc{\tauoAlph}[1]{\ensuremath{\tau_{o,#1}}}
          \nc{\tauoSigma}{\tauoAlph{\Sigma}}

          \nc{\tauu}{\ensuremath{\tau_u}}
          \nc{\tauuAlph}[1]{\ensuremath{\tau_{u,#1}}}
          \nc{\tauuSigma}{\tauuAlph{\Sigma}}

          \nc{\taub}{\ensuremath{\tau_b}}
          \nc{\taubAlph}[1]{\ensuremath{\tau_{b,#1}}}
          \nc{\taubSigma}{\taubAlph{\Sigma}}

          \nc{\Class}{\ensuremath{\mathfrak{C}}}

          \nc{\atoms}{\ensuremath{\textit{atoms}}}

          \nc{\rroot}{\ensuremath{\textit{root}}}

\nc{\IDB}[1]{\ensuremath{\textit{#1}}}
\nc{\ANS}{\IDB{Ans}}
\nc{\WHITE}{\IDB{White}}
\nc{\size}[1]{\ensuremath{|\!|#1|\!|}}

\nc{\Bool}[1]{\ensuremath{#1_{\textit{Bool}}}}
\nc{\Aut}[1]{\ensuremath{\texttt{\upshape #1}}} 
\nc{\ATA}[1]{\ensuremath{\hat{\texttt{\upshape #1}}}} 
\nc{\ATAA}{\ATA{A}}

\nc{\PT}{\ensuremath{\textit{PT}}}
\nc{\myparagraph}[1]{ \smallskip \textbf{#1.} }
\nc{\poly}{\ensuremath{\textup{poly}}}

\nc{\Bin}[1]{\ensuremath{\textit{bin}(#1)}}

\nc{\Proj}{\ensuremath{\textit{proj}}}

\nc{\emptystring}{\ensuremath{\varepsilon}}

\nc{\Up}{\ensuremath{\textit{up}}}
\nc{\Stay}{\ensuremath{\textit{stay}}}
\nc{\DownLeft}{\ensuremath{\textit{left}}}
\nc{\DownRight}{\ensuremath{\textit{right}}}

\nc{\Op}{\ensuremath{\textit{Op}}}

\nc{\KK}{\ensuremath{\mathcal{K}}}
\nc{\DD}{\ensuremath{\mathcal{D}}}

\nc{\TPCT}{\ensuremath{\text{TPCT}}}

\nc{\Relevant}{\ensuremath{\textit{Relevant}}}
\nc{\Candidate}{\ensuremath{\textit{Candidate}}}
\nc{\Column}{\ensuremath{\textit{Column}}}
\nc{\Buggy}{\ensuremath{\textit{Buggy}}}
\nc{\Okay}{\ensuremath{\textit{Okay}}}

\nc{\PPrel}{\ensuremath{\PP_{\Relevant}}}

%% file: main.tex
\begin{abstract}
We show that the query containment problem for monadic datalog on
finite unranked labeled trees can be solved in 2-fold exponential time when
(a) considering unordered trees using the axes \emph{child} and
\emph{descendant}, 
and when (b) considering ordered trees using the axes
\emph{firstchild}, \emph{nextsibling},  \emph{child}, and \emph{descendant}. 
When omitting the \emph{descendant}-axis, we obtain that in both cases
the problem is $\EXPTIME$-complete.
\end{abstract}

\section{Introduction}\label{section:Introduction}

The query containment problem (QCP) is a fundamental problem that has been
studied for various query languages. Datalog is a standard tool for
expressing queries with recursion. From Cosmadakis et al.\ \cite{CGKV}
and Benedikt et al.\ \cite{DBLP:conf/icalp/BenediktBS12} it is known
that the QCP for \emph{monadic} datalog queries on the class of
all finite relational structures is $\TwoEXPTIME$-complete.
Restricting attention to finite unranked labeled trees,
Gottlob and Koch \cite{GottlobKoch}
showed that on \emph{ordered} trees the QCP for monadic datalog is 
$\EXPTIME$-hard and decidable, leaving open the question of a tight bound.

Here we show a matching $\EXPTIME$ upper bound for the QCP for monadic
datalog on ordered trees using the axes
\emph{firstchild}, \emph{nextsibling}, and \emph{child}. When adding
the \emph{descendant}-axis, we obtain a $\TwoEXPTIME$ upper bound. 
This, in particular, also yields a $\TwoEXPTIME$ upper bound for the
QCP for monadic datalog 
on \emph{unordered} trees using the axes \emph{child} and
\emph{descendant}, and an $\EXPTIME$ upper bound for unordered trees
using only the \emph{child}-axis. The former result answers a 
question posed by Abiteboul et al.\ in \cite{AbiteboulBMW13}. We
complement the latter result by a matching lower bound.

The paper is organised as follows. 
Section~\ref{section:Preliminaries} fixes the basic notation concerning
datalog queries, (unordered and ordered) trees and their representations as
logical structures, and summarises basic properties of monadic datalog
on trees.
Section~\ref{section:QCPofMonadicDatalog} presents our main results
regarding the query containment problem for monadic datalog on trees. 
Due to space limitations, most technical details had to be deferred to
the appendix of this paper.

\section{Trees and Monadic Datalog ($\mDatalog$)}\label{section:Preliminaries}\label{section:Datalog}

Throughout this paper, $\Sigma$ will
always denote a finite non-empty alphabet.
\\
By $\NN$ we denote the set of non-negative integers, and we let $\NNpos\deff\NN\setminus\set{0}$.

\myparagraph{Relational Structures}
As usual, a \emph{schema} $\tau$ consists of a
finite number of relation symbols $R$, each of a fixed
\emph{arity} $\ar(R)\in\NNpos$.
A \emph{$\tau$-structure} $\A$ consists of a \emph{finite} non-empty
set $A$ called the \emph{domain}
of $\A$, and a relation
$R^{\A}\subseteq A^{\ar(R)}$ for each relation symbol $R\in\tau$. 
It will often be convenient to
identify $\A$ with the \emph{set of atomic facts of $\A$},
i.e., the set $\atoms(\A)$ consisting of all facts $R(a_1,\ldots,a_{\ar(r)})$
for all relation symbols $R\in\tau$ and all tuples $(a_1,\ldots,a_{\ar(R)})\in R^\A$.

If $\tau$ is a schema and $\ell$ is a list of relation symbols, we
write $\tau^\ell$ to denote the extension of the schema $\tau$ by the
relation symbols in $\ell$.
Furthermore, $\tau_\Sigma$ denotes the extension of $\tau$ by
new unary relation symbols $\Label_\alpha$, for all $\alpha\in\Sigma$.

\myparagraph{Unordered Trees}
An \emph{unordered $\Sigma$-labeled tree} $T=(V^T,\lambda^T,E^T)$
consists of a finite set $V^T$ of nodes, a function  
$\lambda^T: V^T\to \Sigma$ assigning to each node $v$ of $T$ a label
$\lambda(v)\in\Sigma$, and a set $E^T\subseteq V^T\times V^T$ of directed edges such
that the graph $(V^T,E^T)$ is a rooted tree where edges are directed
from the root towards the leaves.
We represent such a tree $T$ as a relational structure of domain $V^T$
with unary and binary relations: For each label $\alpha\in
\Sigma$, $\Label_\alpha(x)$ expresses that $x$ is a node with label
$\alpha$;
$\Child(x,y)$ expresses that $y$ is a child of node $x$; 
$\Root(x)$ expresses that $x$ is the tree's root node;
$\Leaf(x)$ expresses that $x$ is a leaf; and $\Desc(x,y)$ expresses
that $y$ is a descendant of $x$ (i.e., $y$ is a child or a grandchild
or \ldots of $x$).
We denote this relational structure
representing $T$ by $\S_u(T)$, but when no confusion
arises we simply write $T$ instead of $\S_u(T)$.

The queries we consider for unordered trees are allowed to make use of at least the
predicates $\Label_\alpha$ and $\Child$.
We fix the schema
\[
  \tauu \deff \set{\Child}.
\]
The representation of unordered $\Sigma$-labeled trees as $\tauuSigma$-structures was
considered, e.g., in \cite{AbiteboulBMW13}.

\myparagraph{Ordered Trees}
An \emph{ordered} $\Sigma$-labeled tree $T=(V^T,\lambda^T,E^T,\textit{order}^T)$
has the same components as an unordered $\Sigma$-labeled tree and,
in addition, $\textit{order}^T$ fixes for each node $u$ of $T$,
a strict linear order of all the children of $u$ in $T$.

To represent such a tree as a relational structure, we use the same
domain and the same 
predicates as for unordered $\Sigma$-labeled trees, along with three further
predicates $\Fc$ (``first-child''), $\Ns$ (``next-sibling''), and
$\Ls$ (``last sibling''), where $\Fc(x,y)$ expresses that $y$
is the first child of node $x$ (w.r.t.\ the linear order of the
children of $x$ induced by $\textit{order}^T$); $\Ns(x,y)$
expresses that $y$ is the right sibling of $x$ (i.e., $x$ and $y$ have
the same parent $p$, and $y$ is the immediate successor of $x$ in the
linear order of $p$'s children given by $\textit{order}^T$); and 
$\Ls(x)$ expresses that $x$ is the rightmost sibling (w.r.t.\ the
linear order of the children of $x$'s parent given by $\textit{order}^T$).
We denote this relational structure representing $T$ by $\S_o(T)$,
but when no confusion arises we simply write $T$ instead of $\S_o(T)$.

The queries we consider for ordered trees are allowed to make use of
at least the predicates $\Label_\alpha$, $\Fc$, and $\Ns$.
We fix the schemas
\[
  \tauo \deff \set{\Fc, \, \Ns}
  \qquad and \qquad
  \tauGK\deff\tauo^{\Root,\Leaf,\Ls}.
\]
In \cite{GottlobKoch}, Gottlob and Koch represented ordered
$\Sigma$-labeled trees 
as $\tauGKSigma$-structures.

\myparagraph{Datalog}
We assume that the reader is familiar with the syntax and semantics of
\emph{datalog} (cf., e.g., \cite{Datalog-Bible,GottlobKoch}).
Predicates that occur in the head of
some rule of 
a datalog program
$\PP$ are called \emph{intensional}, whereas predicates that only
occur in the body of rules of $\PP$ are called \emph{extensional}.
By $\idb(\PP)$ and
$\edb(\PP)$ we denote the sets of intensional and extensional predicates
of $\PP$, resp. We say that $\PP$ \emph{is of schema $\tau$}
if $\edb(\PP)\subseteq \tau$.
We write $\T_{\PP}$ to denote the \emph{immediate
  consequence operator} associated with a datalog program $\PP$.
Recall that $\T_{\PP}$ maps a set $C$ of atomic facts to the set of all
atomic facts that are derivable from $C$ by at most one
application of the rules of $\PP$ (see e.g.\ \cite{Datalog-Bible,GottlobKoch}). 
The monotonicity of $\T_{\PP}$
implies that for each finite set $C$, the iterated application of
$\T_{\PP}$ to $C$ leads to a fixed point, denoted by
$\T_{\PP}^\omega(C)$, which is reached after a finite number of
iterations.

\myparagraph{Monadic datalog queries}
A datalog program belongs to \emph{monadic datalog}
($\mDatalog$, for short), 
if all its \emph{intensional} predicates have arity~1. 

A \emph{unary monadic datalog query} of schema $\tau$ is a tuple $Q=(\PP,P)$
where $\PP$ is a monadic datalog program of schema $\tau$ and $P$ is an
intensional predicate of $\PP$.
$\PP$ and $P$ are called the \emph{program} and the \emph{query
  predicate} of $Q$.
When evaluated in a finite $\tau$-structure $\A$ that represents a
labeled tree
$T$, the query $Q$ results in the unary relation
\,$\AF{Q}(T) \, \deff \,
  \setc{ a\in A\,}{\,  P(a) \,\in\,
    \T_\PP^\omega(\atoms(\A))\, }$.

The \emph{Boolean monadic datalog query} $\Bool{Q}$ specified by
$Q=(\PP,P)$ is the Boolean query with $\AF{\Bool{Q}}(T)=\Yes$ iff the
tree's root node belongs to $\AF{Q}(T)$.

The \emph{size} $\size{Q}$ of a monadic datalog query $Q$ is the
length of $Q=(\PP,P)$ viewed as a string over a suitable alphabet.

\myparagraph{Expressive power of monadic datalog on trees}
From Gottlob and Koch \cite{GottlobKoch} we know that on \emph{ordered}
$\Sigma$-labeled trees represented as $\tauGKSigma$-structures, 
monadic datalog can express exactly the same unary queries as monadic
second-order logic --- for short, we will say ``$\mDatalog(\tauGK)=\MSO(\tauGK)$ on
ordered trees''. Since the $\Child$ and $\Desc$ relations are
definable in $\MSO(\tauGK)$, this implies that
$\mDatalog(\tauGK)=\mDatalog(\tauGK^{\Child,\Desc})$ on ordered trees.

On the other hand, using the monotonicity of the immediate
consequence operator, 
one obtains
that removing any of the predicates $\Root,
\Leaf,\Ls$ from $\tauGK$ strictly decreases the expressive power of
$\mDatalog$ on ordered trees (see \cite{DBLP:journals/corr/FroSchw13}). 
By a similar reasoning
one also obtains that on \emph{unordered} trees,
represented as $\tauuSigma^{\Root,\Leaf,\Desc}$-structures, monadic
datalog is strictly less expressive than monadic
second-order logic, and omitting any of the predicates $\Root$, $\Leaf$ 
further reduces the expressiveness of
monadic datalog on unordered trees \cite{DBLP:journals/corr/FroSchw13}.

\section{Query Containment for Monadic Datalog on Trees}\label{section:QCPofMonadicDatalog}

Let $\tau_\Sigma$ be one of the schemas introduced in
Section~\ref{section:Preliminaries} for representing (ordered or
unordered) $\Sigma$-labeled trees as relational structures.
For two unary queries $Q_1$ and $Q_2$ of schema $\tau_\Sigma$
we write $Q_1 \subseteq Q_2$ 
to indicate that for every  $\Sigma$-labeled
tree $T$ we have $\AF{Q_1}(T) \subseteq \AF{Q_2}(T)$.
Similarly, if $Q_1$ and $Q_2$ are \emph{Boolean} queries of schema $\tau_\Sigma$, we write
$Q_1\subseteq Q_2$ to indicate that for every $\Sigma$-labeled tree
$T$, if $\AF{Q_1}(T)=\Yes$ then also $\AF{Q_2}(T)=\Yes$.
We write $Q_1\not\subseteq Q_2$ to indicate that
$Q_1\subseteq Q_2$ does not hold. 
The \emph{query containment problem} (QCP, for short) is defined as follows:

\begin{Problem}{The QCP for $\mDatalog(\tau)$ on trees}
	\In A finite alphabet $\Sigma$ and\\ two (unary or Boolean) $\mDatalog(\tau_\Sigma)$-queries $Q_1$ and
        $Q_2$. 
        \Quest Is $Q_1\subseteq Q_2$\,?
\end{Problem}

It is not difficult to see that
this problem is decidable: the first step is to observe that monadic
datalog can effectively be embedded into monadic second-order logic, the second
step then applies the well-known result that the monadic second-order
theory of finite labeled trees is decidable (cf., e.g., \cite{WThomas-Handbook-survey,tata2008}).

Regarding ordered trees represented as $\tauGK$-structures, in
\cite{GottlobKoch} it was shown that the QCP for unary
$\mDatalog(\tauGK)$-queries on trees is $\EXPTIME$-hard.
Our first main result generalises this to unordered trees represented
as $\tau_u$-structures:

\newcounter{Counter_Thm:unorderedExptimeHard}
\setcounter{Counter_Thm:unorderedExptimeHard}{\value{theorem}}

\begin{Theorem}\label{Thm:unorderedExptimeHard}\ \\
 The QCP for Boolean $\mDatalog(\tau_u)$ on unordered trees is $\EXPTIME$-hard.
\end{Theorem}

Our proof proceeds via a reduction from the $\EXPTIME$-complete
\emph{two person corridor tiling} (TPCT) problem
\cite{Chle86}: For a given instance $I$ of the TPCT-problem we
construct (in polynomial time) an alphabet $\Sigma$ and two Boolean $\mDatalog(\tauuSigma)$-queries
$Q_1$, $Q_2$ which enforce that any tree $T$ witnessing that
$Q_1\not\subseteq Q_2$, contains an encoding of a winning strategy for 
the first player
of the TPCT-game associated with~$I$.
Using Theorem~\ref{Thm:unorderedExptimeHard} along with a method of
\cite{GottlobKoch} for replacing the $\Child$-predicate by means of
the predicates $\Fc,\Ns$, we can transfer the hardness result to
ordered trees represented by $\tau_o$-structures:

\newcounter{Counter_Cor:orderedExptimeHard}
\setcounter{Counter_Cor:orderedExptimeHard}{\value{theorem}}

\begin{Corollary}\label{Cor:orderedExptimeHard}\ \\
 The QCP for Boolean $\mDatalog(\tau_o)$ on ordered trees is $\EXPTIME$-hard.
\end{Corollary}

Our second main result provides a matching $\EXPTIME$ upper bound for
the QCP on ordered trees, even in the presence of all predicates in
$\tauGK^{\Child}$:

\newcounter{Counter_Thm:orderedInExptime}
\setcounter{Counter_Thm:orderedInExptime}{\value{theorem}}

\begin{Theorem}\label{Thm:orderedInExptime}\ \\
 The QCP for unary
 $\mDatalog(\tauGK^{\Child})$ on ordered trees belongs to $\EXPTIME$.
\end{Theorem}

\begin{proofsketch} 
Consider a schema $\tau\subseteq \tauGK^{\Child,\Desc}$.
By using the automata-theoretic approach \cite{CGKV},
a canonical method for
deciding the QCP for unary $\mDatalog(\tau)$ proceeds as
follows:
\begin{me}
 \item\label{item:Step1}
   Transform the input queries $Q_1$ and $Q_2$ into \emph{Boolean}
   queries $Q'_1$ and $Q'_2$ on \emph{binary} trees, such that $Q_1\subseteq Q_2$ iff
   $Q'_1\subseteq Q'_2$.\vspace{0.5mm}
 \item\label{item:Step2}
   Construct tree automata $\Aut{A}_1^{\Yes}$ and $\Aut{A}_2^{\No}$ such that
   $\Aut{A}_1^{\Yes}$ (resp.\ $\Aut{A}_2^{\No}$) accepts exactly those trees $T$ with $Q'_1(T)=\Yes$
   (resp.\ $Q'_2(T)=\No$).\vspace{0.5mm}
 \item\label{item:Step3}
   Construct the product automaton $\Aut{B}$ of $\Aut{A}_1^{\Yes}$ and
   $\Aut{A}_2^{\No}$, such that $\Aut{B}$ accepts exactly those trees that are
   accepted by $\Aut{A}_1^{\Yes}$ and by $\Aut{A}_2^{\No}$.
   Afterwards, check if the tree language recognised by $\Aut{B}$ is
   empty.
   Note that this is the case if, and only if, $Q_1\subseteq Q_2$.
\end{me}
Using time polynomial in the size of $Q_1$ and $Q_2$, Step~\eqref{item:Step1} can be achieved in a standard way by
appropriately extending the labelling alphabet $\Sigma$. 

For Step~\eqref{item:Step3}, if $\Aut{A}_1^{\Yes}$ and
$\Aut{A}_2^{\No}$ are nondeterministic bottom-up tree automata, 
the construction of $\Aut{B}$ takes time polynomial in the sizes of
$\Aut{A}_1^{\Yes}$ and $\Aut{A}_2^{\No}$, and the
emptiness test can be done in time polynomial in the size of $\Aut{B}$
(see e.g. \cite{tata2008}).

The first idea for tackling Step~\eqref{item:Step2} is to use a standard translation of
Boolean monadic datalog queries into monadic second-order ($\MSO$) sentences:
It is not difficult to see 
(cf., e.g.\ \cite{GottlobKoch}) that any
Boolean $\mDatalog(\tau)$-query $Q$ can be translated in polynomial time
into an equivalent $\MSO$-sentence $\varphi_{Q}$ of the form
\[\textstyle
  \forall X_1\cdots\forall X_n \ \exists z_1\cdots \exists z_\ell \ 
    \Oder_{j=1}^m \gamma_j
\]
where $n$ is the number of intensional predicates of $Q$'s
monadic datalog program $\PP$,
$\ell$ and $m$ are linear in the size of $Q$, and each $\gamma_j$ is a
conjunction of at most $b$ atoms or negated atoms, where $b$ is linear
in the maximum number of atoms occurring in the body of a rule of $\PP$.
Applying the standard method for translating $\MSO$-sentences into
tree automata (cf., e.g., \cite{WThomas-Handbook-survey}), 
we can translate the sentence $\nicht\varphi_Q$ 
into a nondeterministic bottom-up
tree-automaton $\Aut{A}^{\No}$ that accepts a tree $T$ iff
$\AF{Q}(T)=\No$.
This automaton has
$2^{(m'\cdot c^{b'})}$ states, where $m'$ and $b'$ are linear in $m$ and
$b$, resp., 
and $c$ is a constant not depending on $Q$ or $\Sigma$; and $\Aut{A}^{\No}$ can be
constructed in time polynomial in 
$|\Sigma|{\cdot}2^{n+\ell+ m'{\cdot}c^{b'}}$.

Using the subset construction, one obtains
an automaton $\Aut{A}^{\Yes}$ which accepts
a tree $T$ iff $\AF{Q}(T)=\Yes$; and this automaton has $2^{2^{(m'\cdot c^{b'})}}$ 
states. 

Note that, a priori, $b'$ might be linearly related to the size of
$Q$.
Thus, the approach described so far leads to a \emph{3-fold exponential}
algorithm that solves the QCP for unary $\mDatalog(\tau)$-queries.

In case that $\tau$ does \emph{not} contain the $\Desc$-predicate, we
obtain a 2-fold exponential algorithm as follows: At the end of
Step~\eqref{item:Step1} we rewrite $Q'_1$ and $Q'_2$ into queries that
do not contain the $\Child$-predicate ,
and we transform both queries into \emph{tree
  marking normal form} (TMNF), i.e., a normal form in which bodies of
rules consist of at most two atoms, at least one of which is unary.  
From \cite{GottlobKoch} we obtain that
these transformations can be done in time polynomial in the size of $Q'_1$ and $Q'_2$.
Note that for TMNF-queries, the parameters $b$ and $b'$ are constant
(i.e., they do not depend on the query), and thus the above description shows
that for TMNF-queries the automaton $\Aut{A}_2^{\No}$ can be constructed
in 1-fold exponential time, and $\Aut{A}_1^{\Yes}$ can be constructed in
2-fold exponential time.

Finally, the key idea to obtain a \emph{1-fold} exponential algorithm
solving the QCP is to use a different construction for the
automaton $\Aut{A}_1^{\Yes}$, which does not use the detour via an
$\MSO$-formula but, instead, takes a detour via a \emph{two-way
  alternating tree automaton} (2ATA):
We show that a Boolean TMNF-query can be translated, in polynomial
time, into a 2ATA $\hat{\Aut{A}}_1^{\Yes}$
that accepts a tree $T$ iff $\AF{Q_1}(T)=\Yes$.
It is known that, within 1-fold exponential time, a 2ATA
can be transformed into an equivalent nondeterministic bottom-up tree
automaton (this was claimed already in
\cite{CGKV}; detailed proofs of more general results can be found in \cite{Vardi98,Maneth10TypeChecking}).
In summary, this leads to a 1-fold exponential algorithm for solving
the QCP for $\mDatalog(\tauGK^{\Child})$ on ordered trees.
\end{proofsketch}

Since $\tau_u^{\Root,\Leaf}\subseteq \tauGK^{\Child}$, Theorem~\ref{Thm:orderedInExptime} immediately implies:

\newcounter{Counter_Cor:unorderedInExptime}
\setcounter{Counter_Cor:unorderedInExptime}{\value{theorem}}

\begin{Corollary} \label{Cor:unorderedInExptime}
 The QCP for unary
 $\mDatalog(\tau_u^{\Root,\Leaf})$ on unordered trees belongs to $\EXPTIME$.
\end{Corollary}

It remains open if the $\EXPTIME$-membership results of
Theorem~\ref{Thm:orderedInExptime} and
Corollary~\ref{Cor:unorderedInExptime} can be generalised to queries
that also use the descendant predicate $\Desc$. However, the first approach
described in the proof of Theorem~\ref{Thm:orderedInExptime} yields a 3-fold exponential algorithm.
We can improve this by using methods and results from \cite{GottlobKoch} and
\cite{GottlobKochSchulz} to eliminate the $\Desc$-predicate at the
expense of an exponential blow-up of the query size. Afterwards, we
apply the algorithms provided by Theorem~\ref{Thm:orderedInExptime} and
Corollary~\ref{Cor:unorderedInExptime}. This leads to the following:

\newcounter{Counter_Thm:usingDesc}
\setcounter{Counter_Thm:usingDesc}{\value{theorem}}

\begin{Theorem}\label{Thm:usingDesc}
  The QCP for unary $\mDatalog(\tau_u^{\Root,\Leaf,\Desc})$ on
  unordered trees and for unary $\mDatalog(\tauGK^{\Child,\Desc})$ on
  ordered trees can
  be solved in 2-fold exponential time.
\end{Theorem}

%% file: appendix.tex
\appendix

\section*{APPENDIX}

\newcounter{restoreAppTheorem}

\noindent
This appendix contains technical details which were omitted in the main part of the paper.

\begin{mi}
\item
Appendix~\ref{appendix:SyntaxAndSemantics-mDatalog} contains further
basic notation, including a precise
definition of the syntax and semantics of datalog.
\item
Appendix~\ref{appendix:ProofOfThmUnorderedExptimeHard} 
gives a detailed proof of
Theorem~\ref{Thm:unorderedExptimeHard}.
\item
Appendix~\ref{appendix:ProofOfCorUnrderedExptimeHard} provides a proof of
Corollary~\ref{Cor:orderedExptimeHard}.
\item
Appendix~\ref{appendix:orderedInExptime} gives a detailed proof of
Theorem~\ref{Thm:orderedInExptime}.
\item
Appendix~\ref{appendix:GettingRidOfDesc} presents a proof of Theorem~\ref{Thm:usingDesc}.
\end{mi}

\clearpage

\section{Basic Notation and Syntax and Semantics of Datalog}\label{appendix:SyntaxAndSemantics-mDatalog}

\subsubsection*{Basic notation}
For a set $S$ we write $2^S$ to denote the power set of $S$.

Let $\tau$ be a schema suitable for representing ordered (or
unordered) $\Sigma$-labeled trees. Two $\mDatalog(\tau)$-queries $Q$
and $Q'$ are called \emph{equivalent} if $\AF{Q}(T)=\AF{Q'}(T)$ is
true for all finite ordered (or unordered, resp.) $\Sigma$-labeled trees $T$.
\bigskip

\noindent
The following definition of datalog is basically taken from 
\cite{Datalog-Bible}.

\subsubsection*{Syntax of datalog}
A \textit{datalog rule} is an expression of the form 
\ \( 
  h \leftarrow  b_1 ,\ldots, b_n ,
\) \
for $n\in\NN$,
where $h, b_1 , \ldots, b_n$ are called \emph{atoms} of the rule, $h$
is called the rule's \emph{head}, and $b_1,\ldots,b_n$ (understood as a
conjunction of atoms) is called the \emph{body}.
Each atom is of the form $P(x_1,\ldots,x_m)$ where $P$ is a predicate
of some arity $m\in\NNpos$ and $x_1,\ldots,x_m$ are variables.
Rules are required to be \emph{safe} in the sense that all variables
appearing in the head also have to appear in the body.

A \textit{datalog program} is a finite set of datalog rules. 
Let $\PP$ be a datalog program and let $r$ be a datalog rule. 
We write 
$\Var(r)$ for the set of all variables occurring in the rule $r$, and
we let $\Var(\PP):= \bigcup_{r \in \PP} \Var(r)$. 
Predicates that occur in the head of
some rule of $\PP$ are called \emph{intensional}, whereas predicates that only
occur in the body of rules of $\PP$ are called \emph{extensional}.
We write $\idb(\PP)$ and
$\edb(\PP)$ to denote the sets of intensional and extensional predicates
of $\PP$, respectively. We say that $\PP$ \emph{is of schema $\tau$}
if $\edb(\PP)\subseteq \tau$.

\subsubsection*{Semantics of datalog}
For defining the semantics of datalog, let $\tau$ be a schema,
let $\PP$ be a datalog program of schema $\tau$, let
$A$ be a domain, and let
\begin{eqnarray*}
 F_{\PP,A} & \deff & 
 \setc{\ R(a_1,\ldots,a_r)}{R\in \tau\cup\idb(\PP),\ r=\ar(R),\ a_1,\ldots,a_r\in A
 \ }
\end{eqnarray*}
be the set of all \emph{atomic facts over $A$}. 
A \textit{valuation $\beta$ for $\PP$ in $A$} is a function $\beta: \big({\Var(\PP)}
\cup A\big) \to A$ where $\beta(a)=a$ for all $a\in A$.
For an atom $b\deff P(x_1,\ldots,x_m)$ occurring in a rule of $\PP$ we
let 
\ $\beta(b)\deff
P\big(\beta(x_1),\ldots,\beta(x_m)\big)$.
The \emph{immediate consequence operator} $\T_{\PP}$
induced by $\PP$ on $A$ maps
every $C\subseteq F_{\PP,A}$ to 
\begin{eqnarray*}
 \T_{\PP}(C) & \deff \ &
  C \, \cup \,
 \left\{\ 
    \beta(h) \ : \
    \begin{array}{p{7.5cm}}
     there is 
     a rule $h \leftarrow b_1,\ldots,b_n$ in $\PP$ and
     a valuation $\beta$ for $\PP$ in $A$ such that 
     $\beta(b_1),\ldots,\beta(b_n)\in C$
  \end{array}
 \right\}.
\end{eqnarray*}
Clearly, $\T_{\PP}$ is \emph{monotone}, i.e., 
$\T_{\PP}(C)\subseteq \T_{\PP}(D)$ holds for all
$C\subseteq D\subseteq F_{\PP,A}$.
Letting $\T_{\PP}^0(C)\deff C$ and $\T_{\PP}^{i+1}(C)\deff
\T_{\PP}\big(\T_{\PP}^i(C)\big)$ for all $i\in\NN$,
one obtains
\[
  C = \T_{\PP}^0(C) \ \subseteq\ \T_{\PP}^1(C) \ \subseteq\ \cdots \
  \subseteq \
  \T_{\PP}^{i}(C)\ \subseteq\ \T_{\PP}^{i+1}(C) \ \subseteq \ \cdots \ \subseteq\ F_{\PP,A}.
\]
For a finite domain $A$, the set $F_{\PP,A}$ is finite, and hence
there is an $i_0\in\NN$ such that $\T_{\PP}^{i_0}(C)=\T_{\PP}^i(C)$ for all
$i\geq i_0$. In particular, the set $\T_{\PP}^\omega(C)\deff\T_{\PP}^{i_0}(C)$ is a
\emph{fixpoint} of the 
operator $\T_{\PP}$.
By the theorem of Knaster and Tarski 
we know that this fixpoint is the
\emph{smallest} fixpoint of $\T_{\PP}$ which contains $C$.

\clearpage 

\section{$\EXPTIME$-Hardness: Proof of Theorem~\ref{Thm:unorderedExptimeHard}}\label{appendix:ProofOfThmUnorderedExptimeHard}

The aim of this appendix is to prove the following:

\setcounter{restoreAppTheorem}{\value{theorem}}
\setcounter{theorem}{\value{Counter_Thm:unorderedExptimeHard}}

\begin{Theorem}\textbf{(restated)} \\
 The QCP for Boolean $\mDatalog(\tau_u)$ on unordered trees is $\EXPTIME$-hard.
\end{Theorem}

\setcounter{theorem}{\value{restoreAppTheorem}}

We will show this by first proving the according hardness result for the schema $\tau_u^{\Root,\Leaf}$.
Afterwards, we will construct a polynomial-time reduction which
provides the same hardness result also
for the schema $\tau_u$.

\subsection{$\EXPTIME$-hardness result for the schema $\tau_u^{\Root,\Leaf}$}

This subsection's main result is

\begin{Proposition}\label{Prop:Hardness}\ \\
 The QCP for Boolean $\mDatalog(\tau_u^{\Root,\Leaf})$ on unordered trees is $\EXPTIME$-hard.
\end{Proposition}

\begin{proof}

Our proof proceeds by reduction from the $\EXPTIME$-complete \emph{two person corridor tiling} problem ($\TPCT$)
\cite{Chle86}. The task of the $\TPCT$-problem is to decide whether the first player in the following
\emph{two person corridor tiling game} has a winning strategy.

There are two players: Player~1 (the \emph{Constructor}) and Player~2 (the \emph{Saboteur}). 
The game board is a corridor of a given width $n$ and an unbounded length.
There is a finite set $\DD$ of types of \emph{tiles} (or, \emph{dominoes}), and from every tile type, an 
unlimited number of tiles is available. 
The first row $f$ (of width $n$) of tiles, as well as the designated last row $\ell$ (of width $n$) of tiles are given. 

The players alternately select a tile and put it into the next vacant position
(row-wise from left to right);
Player~1 starts at the leftmost position of the second row.
Both players have to respect horizontal and vertical constraints, given by 
two sets $H,V\subseteq\DD^2$.
A tile $d$ chosen for the $j$-th column of the $i$-th row has to fit to its vertical neighbour 
$d_v$ in the $j$-th column of the $(i{-}1)$-th row in the sense that $(d_v,d)\in V$. 
Furthermore, if $j\geq 2$, then tile $d$ also has to fit to its horizontal neighbour
$d_h$ in the $(j{-}1)$-th column of the $i$-th row in the sense that $(d_h,d)\in H$.
If a player is unable to choose a fitting tile, Player~1 loses and the game ends.

The ultimate goal of Player~1 is to produce a tiling whose last row is $\ell$; in this case 
he wins and the game ends. Player~2 wins if either the game goes on for an infinite number of steps, or 
one of the players gets stuck in a situation where he cannot find a fitting tile.

The \emph{two person corridor tiling problem} ($\TPCT$) is the following decision problem.

\begin{Problem}{$\TPCT$}
 \In A tuple $I=(\DD,H,V,n,f,\ell)$ such that $\DD$ is a finite set, $H,V\subseteq \DD^2$, $n\geq 2$,
     $f,\ell\in \DD^n$.
 \Quest Does Player~1 have a winning strategy in the\\ two person corridor tiling game specified by $I$?
\end{Problem}
\noindent

\begin{Theorem}[Chlebus \cite{Chle86}]\label{Pro:TPCT}
	The problem $\TPCT$ is $\EXPTIME$-complete.  
\end{Theorem}

Note that $\EXPTIME$ is closed under complementation. Thus, for proving 
Proposition~\ref{Prop:Hardness} it
suffices to give a polynomial-time reduction from $\TPCT$ to the \emph{complement} of the 
QCP for $\mDatalog(\tauu^{\Root,\Leaf})$
on unordered trees.
For a given $\TPCT$-instance $I=(\DD,H,V,n,f,\ell)$ we will construct a finite alphabet
$\Sigma$ and two Boolean $\mDatalog(\tauuSigma^{\Root,\Leaf})$-queries $Q_1,Q_2$, such that 
\[
\begin{array}{ll}
 & \text{Player~1 has a winning strategy in the}
\\
 & \text{two person corridor tiling game specified by $I$} 
\\[1ex]
   \iff \ 
 & \text{there exists an unordered $\Sigma$-labeled tree $T$ such that}
\\
 & \text{$\AF{Q_1}(T)=\Yes$  and  $\AF{Q_2}(T)=\No$ 
     \ (i.e., $Q_1\not\subseteq Q_2$)}.
\end{array}
\]
We will represent strategies for Player~1 by $\Sigma$-labeled trees. 
The query $Q_1$ will describe ``necessary properties'' which are met by every tree that describes a
winning strategy for Player~1, but also by some other trees.
The query $Q_2$ will describe certain ``forbidden properties'' such that a tree which has these
properties for sure does \emph{not} describe a winning strategy for Player~1.

The following representation of a winning strategy for Player~1 is basically taken from
\cite{Loeding-survey}.
We represent a strategy for Player~1 by an unordered $\Sigma$-labeled tree with
\[
  \Sigma\ \deff \ \ \DD\times\set{1,2,\bot,!}.
\]
The first component of a letter $(d,i)\in\Sigma$ corresponds to the tile $d$
that has been played, while the second component indicates whose turn it is to place the next tile
(1 for Player~1, 2 for Player~2, 
$\bot$ in case the game is over because a vertical or horizontal constraint was violated, 
and ! in case that the game is over because Player~1 has won).
In the following, we will say that a node is labeled $d$ (for some $d\in\DD$) to express that its label 
belongs to $\set{d}\times\set{1,2,\bot,!}$. Accordingly, we will say that a node is labeled 
$i$ (for some $i\in\set{1,2,\bot,!}$) to express that its label belongs to $\DD\times\set{i}$.

A finite $\Sigma$-labeled tree $T$ is called \emph{good} if it satisfies the following conditions (1)--(9).
It is not difficult
to verify that Player~1 has a winning strategy if, and only if, 
there exists a finite $\Sigma$-labeled tree that is good.
\begin{me}
\item[(1)]
 The root is labeled by $(d,2)$ for some $d\in\DD$. (This indicates that at the beginning of the game, 
 Player~1 chooses tile $d$, and Player~2 is the one to play in the next step).
\item[(2)]
 Nodes with labels $\bot$ or ! are leaves.
\item[(3)]
 Nodes with labels in $\DD\times\set{1}$ have at least one child. (Such a child describes the choice made by 
 Player~1 in the next step).
\item[(4)]
 Nodes with labels in $\DD\times\set{2}$ have at least $|\DD|$ children --- one for each tile type $d\in\DD$. (These
 children represent the potential choices that Player~2 might make in the next step).
\item[(5)]
 There is no node labeled 1 or 2 such that all of its children are labeled by $\bot$.
 (I.e., the game never gets stuck).
\item[(6)]
 Labels from $\DD\times\set{1}$ and $\DD\times\set{2}$ alternate on each path from the root to a leaf.
 (I.e., both players alternately choose a tile).
\item[(7)]
 If a node $x$ is labeled !, then the number of nodes visited by the path from the root to $x$ is a multiple of $n$ 
 and the last $n$ nodes on this path are labeled according to $\ell$.
 (This means that the last $n$ nodes of the path describe a row which has the desired labeling $\ell$.)
\item[(8)]
 At each node $x$ labeled $(d,i)$ with $i\neq \bot$, the tile $d$ respects the horizontal and the vertical constraints.
\item[(9)]
 At each node labeled $(d,\bot)$ for some $d\in\DD$, the tile $d$ violates the horizontal or the vertical constraints.
\end{me}
To be precise, the conditions (8) and (9) mean the following.
We define the \emph{depth} of a node as follows: The root has depth 1; and for each node $x$ of depth $j$, 
all children of $x$ are of depth $j{+}1$.
\begin{mea}
 \item
  A node $x$ labeled with tile $d\in\DD$ \emph{respects the horizontal constraints} if $x$  
 \begin{mi}
  \item
   is either of depth congruent 1 modulo $n$ (and thus corresponds to a position in the 1-st column of a row),
  \item
   or we have $(d_h,d)\in H$, where the parent of $x$ is labeled with tile $d_h\in\DD$
   (i.e., $x$ corresponds to a position where tile $d$ is chosen in some column $j\geq 2$, and this tile fits
   to its horizontal neighbour $d_h$ in column $j{-}1$).
 \end{mi}
 \item
 A node $x$ labeled with a tile $d\in\DD$ \emph{respects the vertical constraints} if $x$ 
 \begin{mi}
  \item
   is either is of depth $j\in\set{1,\ldots,n}$ and we have $(f_j,d)\in H$\\
   (i.e., $x$ corresponds to the $j$-th position in the second row
   and fits to the $j$-th entry $f_j$ of the first row $f$),
  \item
   or it is of depth $j\geq n{+}1$ and we have $(d_v,d)\in V$, where the ancestor of $x$ at depth $j{-}n$ is
   labeled with tile $d_v\in\DD$\\ (i.e., $x$ corresponds to a position where tile $d$ is chosen in some row $i\geq 3$, 
   and this tile fits to its vertical neighbour $d_v$ in row $i{-}1$).
 \end{mi}
\end{mea}

\noindent
As noted above, Player~1 has a winning strategy if, and only if, there exists a finite $\Sigma$-labeled tree
$T$ that is good, i.e., that satisfies the conditions (1)--(9).
The first idea towards completing the proof of Proposition~\ref{Prop:Hardness} is to try to find 
monadic Datalog queries $Q_1$ and $Q_2$ such that for any $\Sigma$-labeled tree $T$ the following is true:
$T$ is good if, and only if, $\AF{Q_1}(T)=\Yes$ and $\AF{Q_2}(T)=\No$. 
In fact, 
it is not difficult to construct for each condition ($c$) with $c\neq 4$ and 
$c\neq 5$ a Boolean $\mDatalog(\tauu^{\Root,\Leaf})$-query $Q^c$ such that for any $\Sigma$-labeled tree $T$ we have:
\[
   \AF{Q^c}(T)=\Yes \ \iff \ T \text{ violates condition ($c$)}.
\]
However, for the conditions (4) and (5), we were unable to find according monadic datalog queries which
precisely characterise all trees that violate (or all trees that fulfill) these conditions.

As a remedy, we define a notion of \emph{almost-good} trees in such a way that the following is true:
\begin{mei}
 \item[(i)]
   Every \emph{almost-good} tree $T$ contains a \emph{good} tree; \\
   and every \emph{good} tree also is \emph{almost-good}.
 \item[(ii)]
   We can find Boolean $\mDatalog(\tauuSigma^{\Root,\Leaf})$-queries $Q_1, Q_2$ such that for any 
   $\Sigma$-labeled tree $T$ the following is true: $T$ is \emph{almost-good} if, and only if,
   $\AF{Q_1}(T)=\Yes$ and $\AF{Q_2}(T)=\No$.
\end{mei}
For defining the notion of \emph{almost-good} trees, we need the following notation.
Let $T$ be an unordered $\Sigma$-labeled tree. By performing a bottom-up scan of $T$, we 
define the set of nodes that are \emph{candidates} as follows:
\begin{mi}
 \item
   Every \emph{leaf} of $T$ 
   that is labeled $\bot$ or !
   is a \emph{candidate}.
 \item
   For each node $x$ of $T$ that is labeled 1, $x$ is a \emph{candidate} if $x$ has a
   child that is a \emph{candidate} and that is \emph{not} labeled $\bot$.
 \item
   For each node $x$ of $T$ that is labeled 2, $x$ is a \emph{candidate} if 
   \begin{mi} 
    \item
      for each $d\in\DD$, $x$ has a child that is a \emph{candidate} and that is labeled $d$,
    \item
      and $x$ has child that is a \emph{candidate} and that is \emph{not} labeled $\bot$.
   \end{mi}
\end{mi}
Now, we perform a top-down scan of $T$ to define the set of nodes that are \emph{relevant} as 
follows:
\begin{mi}
 \item
   The root of $T$ is \emph{relevant} if it is labeled in $\DD\times\set{2}$ and it is a \emph{candidate}.
 \item
   For each non-root node $x$ of $T$, $x$ is \emph{relevant} if it is a \emph{candidate} and its 
   parent is \emph{relevant}.
\end{mi}
Note that according to this definition, in particular, the following is true:
\begin{mi}
 \item
   Every \emph{relevant} node of $T$ either is a leaf of $T$ or has a child that is \emph{relevant}.
 \item
   If the root of $T$ is \emph{relevant}, then it is labeled in $\DD\times\set{2}$, and
   the set of all \emph{relevant} nodes of $T$ forms
   a tree, which we will call $T_{\Relevant}$.
 \item
   Relevant nodes with labels $\bot$ or ! are leaves.
 \item
   Every \emph{relevant} node with label in $\DD\times\set{1}$ has a \emph{relevant} child
   that is \emph{not}~labeled $\bot$.
 \item
   Every \emph{relevant} node with label in $\DD\times\set{2}$ has, for each $d\in\DD$, a 
   \emph{relevant} child labeled $d$; and it has a \emph{relevant} child that is 
   \emph{not} labeled $\bot$.
\end{mi}
Thus, the following is true for every $\Sigma$-labeled tree $T$:
\begin{center}
 $(*)$: \ \
 \parbox[t]{8cm}{
   If the root of $T$ is \emph{relevant}, then\\
   the tree $T_{\Relevant}$ satisfies the conditions (1)--(5).
 }
\end{center}
Furthermore, note that if $T$ is \emph{good}, then $T_{\Relevant}=T$.

We say that a $\Sigma$-labeled tree $T$ is \emph{almost-good} if its root node is \emph{relevant} and
the tree $T_{\Relevant}$ is \emph{good}, i.e., satisfies the conditions (6)--(9).

Our next goal is to construct an $\mDatalog(\taubSigma^{\Root,\Leaf})$-program $\PPrel$ which constructs, 
in an intensional predicate called $\Relevant$, the set of all \emph{relevant} nodes.
We start with $\PPrel\deff\emptyset$.
To access the parts $d$ and $i$ of a node-label $(d,i)\in\Sigma$, it will be convenient to include into $\PPrel$
the rules
\[
 \Label_d(x) \leftarrow \Label_{(d,i)}(x) 
 \qquad\text{and}\qquad
 \Label_i(x) \leftarrow \Label_{(d,i)}(x)
\]
for every letter $(d,i)\in\Sigma$.
Furthermore, for all $d,d'\in\DD$ and $i,i'\in\set{1,2,\bot,!}$ with $d\neq d'$ and $i\neq i'$ we
add to $\PPrel$ the rules
\[
 \Label_{\neq d}(x) \leftarrow \Label_{d'}(x)
 \qquad\text{and}\qquad
 \Label_{\neq i}(x) \leftarrow \Label_{i'}(x).
\]
To describe the \emph{candidate} nodes, we add to $\PPrel$ the rules
\begin{align*}
 \Candidate(x)&\leftarrow \Leaf(x),\Label_{\bot}(x)
\\
 \Candidate(x)&\leftarrow \Leaf(x),\Label_{\,!}(x)
\\
 \Candidate(x)&\leftarrow \Label_1(x),\Child(x,y),\Candidate(y),\Label_{\neq \bot}(y),
\end{align*}
as well as the following rule, where $d_1,\ldots,d_m$ is a list of all elements in $\DD$:
\begin{align*}
 \Candidate(x) \leftarrow & \Label_2(x),\Child(x,y_1),\ldots,\Child(x,y_m),
\\
 &  \Candidate(y_1),\ldots,\Candidate(y_m),
\\
 &  \Label_{d_1}(y_1),\ldots,\Label_{d_m}(y_m),
\\
 &  \Child(x,y),\Candidate(y),\Label_{\neq \bot}(y).
\end{align*}
To describe the \emph{relevant} nodes, we add to $\PPrel$ the rules
\begin{align*}
 \Relevant(x)&\leftarrow \Root(x),\Candidate(x),\Label_{2}(x)
\\
 \Relevant(x)&\leftarrow \Candidate(x),\Child(y,x),\Relevant(y)
\end{align*}
This completes the definition of the monadic datalog program $\PPrel$.
\\
Obviously, the following is true:
\begin{Claim}\label{claim:Prel-HardnessProof}
 $\PPrel$ can be constructed in time polynomial in the size of $\Sigma$. 
 \\
 Furthermore, for the unary query $Q_{\Relevant}\deff (\PPrel,\Relevant)$ the following is true: 
 For every unordered $\Sigma$-labeled tree $T$, the set
 $\AF{Q_{\Relevant}}(T)$ contains exactly those nodes of $T$ that are relevant.
\end{Claim}

Recall that our overall goal is to find Boolean queries $Q_1$ and $Q_2$ that satisfy condition (ii).
We choose $Q_1$ to be the query that returns ``$\Yes$'' exactly for those trees $T$ 
whose root is \emph{relevant}. I.e., the program of $Q_1$ is obtained from $\PPrel$ by adding the
rule
\begin{align*}
 \Accept(x)&\leftarrow \Root(x),\Relevant(x)
\end{align*}
and the query predicate of $Q_1$ is the predicate $\Accept$.
From $(*)$ we know that the following is true:
\begin{Claim}\label{claim:Q1-HardnessProof}
 $Q_1$ can be constructed in time polynomial in the size of $\Sigma$; and for every 
 $\Sigma$-labeled tree $T$ we have
 $\AF{Q_1}(T)=\Yes$ if, and only if, the root of $T$ is relevant and the tree $T_{\Relevant}$ satisfies the conditions
 (1)--(5).
\end{Claim}

Our next goal is to construct a Boolean query $Q_2$ that returns ``$\Yes$'' exactly for those trees $T$
where the tree $T_{\Relevant}$ violates one of the conditions (6)--(9). 
Once we have achieved this, we know that for any tree $T$ the following is true: 
$\AF{Q_1}(T)=\Yes$ and $\AF{Q_2}(T)=\No$ if,
and only if, the tree $T_{\Relevant}$ satisfies the conditions (1)--(9), and hence witnesses that 
Player~1 has a winning strategy for the two person corridor tiling game specified by 
$I=(\DD,H,V,n,f,\ell)$.

To construct $Q_2$, we start with the monadic Datalog program $\PP_2\deff\PPrel$ and successively add
rules to $\PP_2$.

To detect a violation of condition (6), we add to $\PP_2$ the rules
\begin{align*}
 \Reject^{(6)}(z) \leftarrow & \Relevant(x),\Relevant(y),\\ & \Child(x,y),\Label_1(x),\Label_1(y),\Root(z) 
\\
 \Reject^{(6)}(z) \leftarrow & \Relevant(x),\Relevant(y),\\ & \Child(x,y),\Label_2(x),\Label_2(y),\Root(z).
\end{align*}
This way, $T_{\Relevant}$ violates condition (6) if, and only if, the root of $T$ 
gets assigned the predicate $\Reject^{(6)}$. Thus, the Boolean query specified by $(\PP_2,\Reject^{(6)})$ 
returns ``$\Yes$'' for exactly those trees $T$ where $T_{\Relevant}$ violates condition (6).

To detect a violation of the conditions (7)--(9), it will be convenient to use predicates 
$\Column_j$ for each $j\in\set{1,\ldots,n}$, such that $\Column_j(x)$ indicates that node $x$ corresponds
to a tile placed in column $j$ of the corridor.
Thus, we add to $\PP_2$ the rules
\begin{align*}
  \Column_1(x)&\leftarrow\Root(x),\Relevant(x)
\\
  \Column_1(x)&\leftarrow\Child(y,x),\Column_n(y),\Relevant(y),\Relevant(x)
\end{align*}
and for each $j\in\set{2,\ldots,n}$ the rule
\begin{align*}
  \Column_j(x)\leftarrow\Child(y,x),\Column_{j-1}(y),\Relevant(y),\Relevant(x).
\end{align*}
Furthermore, for each $j\in\set{1,\ldots,n{-}1}$ we add to $\PP_2$ the rule
\begin{align*}
 \Column_{\neq n}(x)&\leftarrow \Column_j(x)
\end{align*}
and for each $j\in\set{2,\ldots,n}$ we add to $\PP_2$ the rule
\begin{align*}
 \Column_{\neq 1}(x)&\leftarrow \Column_j(x).
\end{align*}

To detect a violation of condition (7), we add to $\PP_2$ the rule
\[
 \Reject^{(7)}(z) \leftarrow \Label_{\,!}(x),\Column_{\neq n}(x),\Relevant(x),\Root(z)
\]
and for each $j\in\set{1,\ldots,n}$ we add the rule
\begin{align*}
 \Reject^{(7)}(z)\leftarrow &
 \Label_{\,!}(x_n),\Child(x_1,x_2),\ldots,\Child(x_{n-1},x_n),
\\
 & \Label_{\neq \ell_j}(x_j),\Relevant(x_1),\ldots,\Relevant(x_n),\Root(z)
\end{align*}
where $\ell_j$ denotes the $j$-th position of the designated last row $\ell$.
\\
This way, $T_{\Relevant}$ violates condition (7) if, and only if, the root of $T$ gets assigned the predicate $\Reject^{(7)}$.
Hence, the Boolean query specified by $(\PP_2,\Reject^{(7)})$ returns ``$\Yes$'' for exactly those trees $T$
where $T_{\Relevant}$ violates condition (7).
Note that $\PP_2$ can be constructed in time polynomial in the size of $\Sigma$ and $n$.

To detect a violation of condition (8), it will be convenient to use predicates
$\Buggy_H$ and $\Buggy_V$, such that $\Buggy_H(x)$ (resp., $\Buggy_V(x)$) indicates that node $x$ violates
the horizontal (resp., the vertical) constraints.
Thus, for all $(d_h,d)\in \DD^2\setminus H$, we add to $\PP_2$ the rule
\begin{align*}
  \Buggy_H(x) \leftarrow 
& \Column_{\neq 1}(x), \Child(y,x), \Label_{d_h}(y), \Label_d(x),
\\
& \Relevant(y),\Relevant(x).
\end{align*}
Similarly, for all $(d_v,d)\in \DD^2\setminus V$, we add add to $\PP_2$ the rule
\begin{align*}
  \Buggy_V(x) \leftarrow 
& \Child(y_1,y_2),\ldots,\Child(y_{n-1},y_n),\Child(y_n,x),
\\
& \Label_{d_v}(y_1), \Label_d(x),
\\
& \Relevant(y_1),\ldots,\Relevant(y_n),\Relevant(x).
\end{align*}
To detect nodes that correspond to tiles placed in the corridor's second row, i.e., tiles that must fit to 
the given first row $f=(f_1,\ldots,f_n)\in\DD^n$, we furthermore add for each $j\in\set{1,\ldots,n}$ 
and each $d\in\DD$ with $(f_j,d)\not\in V$, the rule
\begin{align*}
  \Buggy_V(x_j) \leftarrow
& \Root(x_1),\Child(x_1,x_2),\ldots,\Child(x_{n-1},x_n),
\\
&  \Label_d(x_j),\Relevant(x_1),\ldots,\Relevant(x_n)
\end{align*}
To detect a violation of condition (8) we add to $\PP_2$ the rules
\begin{align*}
  \Reject^{(8)}(z) & \leftarrow \Label_{\neq \bot}(x), \Buggy_H(x), \Relevant(x),\Root(z)
\\
  \Reject^{(8)}(z) & \leftarrow \Label_{\neq \bot}(x), \Buggy_V(x), \Relevant(x),\Root(z).
\end{align*}
This way, $T_{\Relevant}$ violates condition (8) if, and only if, the root of $T$ gets assigned the predicate $\Reject^{(8)}$.
Hence, the Boolean query specified by $(\PP_2,\Reject^{(8)})$ returns ``$\Yes$'' for exactly those trees $T$
where $T_{\Relevant}$ violates condition (8).
Note that $\PP_2$ can be constructed in time polynomial in the size of $\Sigma,n,\DD$.

To detect a violation of condition (9), it will be convenient to use predicates
$\Okay_H$ and $\Okay_V$, such that $\Okay_H(x)$ (resp., $\Okay_V(x)$) indicates that node $x$ 
satisfies the horizontal (resp., the vertical) constraints.
Thus, for all $(d_h,d)\in H$, we add to $\PP_2$ the rules
\begin{align*}
  \Okay_H(x) \leftarrow & \Column_{1}(x)
\\
  \Okay_H(x) \leftarrow & \Column_{\neq 1}(x), \Child(y,x), \Label_{d_h}(y), \Label_d(x).
\end{align*}
Similarly, for all $(d_v,d)\in V$, we add add to $\PP_2$ the rules
\begin{align*}
  \Okay_V(x) \leftarrow 
& \Child(y_1,y_2),\ldots,\Child(y_{n-1},y_n),\Child(y_n,x),
\\
& \Label_{d_v}(y_1), \Label_d(x).
\end{align*}
To detect nodes that correspond to tiles placed in the corridor's second row, i.e., tiles that must fit to 
the given first row $f=(f_1,\ldots,f_n)\in\DD^n$, we furthermore add for each $j\in\set{1,\ldots,n}$ 
and each $d\in\DD$ with $(f_j,d)\in V$, the rule
\begin{align*}
  \Okay_V(x_j) \leftarrow
& \Root(x_1),\Child(x_1,x_2),\ldots,\Child(x_{n-1},x_n),\Label_d(x_j).
\end{align*}
To detect a violation of condition (9) we add to $\PP_2$ the rule
\[
  \Reject^{(9)}(z)  \leftarrow \Label_{\bot}(x), \Okay_H(x), \Okay_V(x), \Relevant(x), \Root(z).
\]
This way, $T_{\Relevant}$ violates condition (9) if, and only if, the root of $T$ gets assigned the predicate $\Reject^{(9)}$.
Hence, the Boolean query specified by $(\PP_2,\Reject^{(9)})$ returns ``$\Yes$'' for exactly those trees $T$
where $T_{\Relevant}$ violates condition (9).
Note that $\PP_2$ can be constructed in time polynomial in the size of $\Sigma,n,\DD,H,V$.

Finally, for each $c\in\set{6,7,8,9}$ we add to $\PP_2$ the rule
\[
 \Reject(z)\leftarrow \Reject^{(c)}(z)
\]
and we let $Q_2$ be the Boolean monadic datalog query specified by $(\PP_2,\Reject)$.
By our construction, the following holds:

\begin{Claim}\label{claim:Q2-HardnessProof}
 $Q_2$ can be constructed in time polynommial in the size of $\Sigma,n,\DD,H,V$; and for
every $\Sigma$-labeled tree $T$ we have $\AF{Q_2}(T)=\Yes$ if, and only if, the tree
$T_{\Relevant}$ violates one of the conditions (6)--(9).
\end{Claim}

In summary, for each $\TPCT$-instance 
$I=(\DD,H,V,n,f,\ell)$, we can construct within polynomial time the alphabet $\Sigma\deff \DD\times\set{1,2,\bot,!}$ and
two Boolean $\mDatalog(\taubSigma^{\Root,\Leaf})$-queries $Q_1,Q_2$ such that the following is true 
for every unordered $\Sigma$-labeled tree $T$:
\[
 \begin{array}{ll}
 &  \AF{Q_1}(T)=\Yes \text{ \ and \ } \AF{Q_2}(T)=\No
 \\[1ex]
  \iff \ 
 & 
  \parbox[t]{9cm}{the root of $T$ is \emph{relevant} and \\ the tree
  $T_{\Relevant}$ satisfies the conditions (1)--(9).}
 \end{array}
\]
Thus, $Q_1\not\subseteq Q_2$ if, and only if, Player~1 has a winning strategy in the
two person corridor tiling game specified by $I$.
Hence, we have established a polynomial-time reduction from $\TPCT$ to the complement of the
QCP for Boolean $\mDatalog(\tauu)$ on unordered trees.
This completes the proof of Proposition~\ref{Prop:Hardness}.
\end{proof}

\subsection{Omitting the predicates $\Root$ and $\Leaf$: Proof of Theorem~\ref{Thm:unorderedExptimeHard}}

From Proposition~\ref{Prop:Hardness} we already know that the QCP is $\EXPTIME$-hard for
Boolean $\mDatalog(\tauu^{\Root,\Leaf})$-queries on unordered trees.
Theorem~\ref{Thm:unorderedExptimeHard} claims the same hardness result already for queries that don't use the 
predicates $\Root$ and $\Leaf$. Thus, 
Theorem~\ref{Thm:unorderedExptimeHard} is an immediate consequence of
Proposition~\ref{Prop:Hardness} and the following lemma:

\begin{Lemma}\label{Lemma:GettingRidOfRootLeaf}
There is a polynomial-time reduction from 
the QCP for Boolean $\mDatalog(\tauu^{\Root,\Leaf})$ on unordered trees to
the QCP for Boolean $\mDatalog(\tauu)$ on unordered trees.
\end{Lemma}
\begin{proof}
 Let $\Sigma,Q_1,Q_2$ be an input for the QCP for $\mDatalog(\tauu^{\Root,\Leaf})$ on unordered trees.
 Our goal is to construct, within polynomial time, an alphabet $\tilde{\Sigma}$ and two Boolean
 $\mDatalog(\tauuAlph{\tilde{\Sigma}})$-queries $\tilde{Q}_1, \tilde{Q}_2$, such that 
 $Q_1\subseteq Q_2$ iff $\tilde{Q}_1\subseteq \tilde{Q}_2$.

 We choose $\tilde{\Sigma}\deff\Sigma\times 2^{\set{\Root,\Leaf}}$.
 With every $\Sigma$-labeled tree $T$ we associate the $\tilde{\Sigma}$-labeled tree $\tilde{T}$ that is
 obtained from $T$ by replacing the label of each node $\alpha\in\Sigma$ with the label $(\alpha,I)$ where
 $I\subseteq\set{\Root,\Leaf}$ is given as follows:
 \begin{eqnarray*}
   \Root \in I & \iff & \text{$v$ is the root of $T$},
 \\
   \Leaf \in I & \iff & \text{$v$ is a leaf of $T$}.
 \end{eqnarray*}
 Let $\tilde{\PP}_{\textit{labels}}$ be the $\mDatalog(\tauuAlph{\tilde{\Sigma}})$-program consisting of
 the rules
 \begin{align*}
   \Label_\alpha(x) & \leftarrow \Label_{(\alpha,I)}(x)
 \\
   \Root(x) &\leftarrow \Label_{(\alpha,I')}(x) 
 \\
   \Leaf(x)&\leftarrow \Label_{(\alpha,I'')}(x)
 \end{align*}
 for all $\alpha\in\Sigma$ and all $I,I',I''\subseteq\set{\Root,\Leaf}$
 with $\Root\in I'$ and $\Leaf\in I''$.

 Let $\tilde{\PP}_{\textit{incons}}$ be the $\mDatalog(\tauuAlph{\tilde{\Sigma}})$-program 
 consisting of the rules of $\tilde{\PP}_{\textit{labels}}$, along with the following rules:
 \begin{align*}
   P_{\textit{incons}}(x) & \leftarrow \Root(x), \Child(y,x)
 \\
   P_{\textit{incons}}(x) & \leftarrow \Leaf(x), \Child(x,y)
 \\
   P_{\textit{incons}}(x) & \leftarrow \Child(x,y), P_{\textit{incons}}(y).
 \end{align*}
 The Boolean query $\tilde{Q}_{\textit{incons}}=(\tilde{\PP}_{\textit{incons}},P_{\textit{incons}})$ describes
 all $\tilde{\Sigma}$-labeled trees that are \emph{inconsistent} in the sense that for any
 $\tilde{\Sigma}$-labeled tree $T'$ the following is true:
 \[
   \AF{\tilde{Q}_{\textit{incons}}}(T')=\Yes \ \iff \ 
   \text{there is no $\Sigma$-labeled tree $T$ with $T'=\tilde{T}$}.
 \]
 Now, for the given $\mDatalog(\taubSigma)$-queries $Q_1=(\PP_1,P_1)$ and $Q_2=(\PP_2,P_2)$, we
 choose the $\mDatalog(\taubAlph{\tilde{\Sigma}})$-queries 
 $\tilde{Q}_1=(\tilde{\PP}_1,P_1)$ and $\tilde{Q}_2=(\tilde{\PP}_2,P_\textit{acc})$ as follows:
 \begin{align*}
    \tilde{\PP}_1 \deff \ \ & \tilde{\PP}_{\textit{labels}} \ \cup \ \PP_1,
 \\
    \tilde{\PP}_2 \deff \ \ & \tilde{\PP}_{\textit{incons}} \cup \ \PP_2 \ \cup \
    \set{\ P_{\textit{acc}}(x)\leftarrow P_{\textit{incons}}(x), \ \
           P_{\textit{acc}}(x)\leftarrow P_{2}(x)  \ }.
 \end{align*}
 We claim that $Q_1\not\subseteq Q_2 \iff \tilde{Q}_1\not\subseteq \tilde{Q}_2$.

 For the direction ``$\Longrightarrow$'' let $T$ be a $\Sigma$-labeled tree with 
 $\AF{Q_1}(T)=\Yes$ and $\AF{Q_2}(T)=\No$. Then, clearly, also $\AF{\tilde{Q}_1}(\tilde{T})=\Yes$ and
 $\AF{\tilde{Q}_2}(\tilde{T})=\No$. Thus, $\tilde{Q}_1\not\subseteq \tilde{Q}_2$.

 For the direction ``$\Longleftarrow$'' let $T'$ be a $\tilde{\Sigma}$-labeled tree with
 $\AF{\tilde{Q}_1}(T')=\Yes$ and $\AF{\tilde{Q}_2}(T')=\No$.
 The latter implies that $T'$ is \emph{not} inconistent. Hence, there exists a $\Sigma$-labeled tree $T$
 such that $T'=\tilde{T}$. For this tree we know that $\AF{\tilde{Q}_1}(\tilde{T})=\Yes$ and 
 $\AF{\tilde{Q}_2}(\tilde{T})=\No$. Hence, also
 $\AF{Q_1}(T)=\Yes$ and $\AF{Q_2}(T)=\No$. Thus, $Q_1\not\subseteq Q_2$.
 This completes the proof of Lemma~\ref{Lemma:GettingRidOfRootLeaf}.
\end{proof}

\clearpage

\section{$\EXPTIME$-Hardness: Proof of Corollary~\ref{Cor:orderedExptimeHard}}\label{appendix:ProofOfCorUnrderedExptimeHard}

The aim of this appendix is to prove the following:

\setcounter{restoreAppTheorem}{\value{theorem}}
\setcounter{theorem}{\value{Counter_Cor:orderedExptimeHard}}

\begin{Corollary} \textbf{(restated)}
 The QCP for Boolean $\mDatalog(\tau_o)$-queries on ordered trees is $\EXPTIME$-hard.
\end{Corollary}
\setcounter{theorem}{\value{restoreAppTheorem}}

The proof is via a polynomial-time reduction from the QCP for Boolean
\linebreak[4]
$\mDatalog(\tau_u)$-queries over unordered trees which, according to
Theorem~\ref{Thm:unorderedExptimeHard}, is
$\EXPTIME$-hard.

For establishing the reduction, we will rewrite monadic datalog
programs of schema $\tau_u$ into suitable programs of schema $\tau_o$
(i.e., we will rewrite the $\Child$ relation by means of the
relations $\Fc$ and $\Ns$).
For doing this, we can use a result by Gottlob and Koch
\cite{GottlobKoch} which transforms monadic datalog programs into a
certain normal form called
\emph{Tree-Marking Normal Form} (TMNF).
We will use this normal form also later on, in
Appendix~\ref{appendix:orderedInExptime} and
Appendix~\ref{appendix:GettingRidOfDesc}.

\begin{Def} Let $\tau$ be a schema that consists of relation symbols
  of arity at most 2.
A monadic datalog program $\PP$ of schema $\tau$ is in TMNF if each
rule of $\PP$ is of one of the following forms:\footnote{Gottlob
  and Koch \cite{GottlobKoch} also allow rules of the form
  $X(x)\leftarrow Y(x)$. Note that such a rule is equivalent to the
  rule $X(x)\leftarrow Y(x),Y(x)$.}
\begin{multicols}{2} \begin{mei} 
	\item \label{riii}$X(x)\leftarrow R(x,y),Y(y)$
	\item \label{rii} $X(x)\leftarrow R(y,x),Y(y)$
	\item \label{riv} $X(x)\leftarrow Y(x), Z(x)$
\end{mei}
\end{multicols}
\noindent
where $R$ is a binary predicate from $\tau$, $X\in\idb(\PP)$, and
the unary predicates $Y$ and $Z$ are either intensional or belong to $\tau$. 
\end{Def}

\begin{Theorem}[{Gottlob and Koch 
  \cite[Theorem~5.2]{GottlobKoch}}]\label{Theorem:GK_TMNF}
\ \\
For each monadic datalog program $\PP$ of schema $\tauGK^{\Child}$,
there is an equivalent program in TMNF of schema $\tauGK$, which can be computed in
time $O(\size{\PP})$.
\end{Theorem}

\noindent
A detailed analysis shows that the proof given in
\cite{GottlobKoch} in fact also proves the following:

\begin{Corollary}[implicit in \cite{GottlobKoch}]\label{remark_1}
   For each monadic datalog program $\PP$ of schema $\tau_o^{\Child}$,
   there is an equivalent program in TMNF of schema $\tau_o$, which can be computed in
   time $O(\size{\PP})$.
\end{Corollary}

\noindent
We are now ready for the proof of Corollary~\ref{Cor:orderedExptimeHard}.

\begin{proofof}{Corollary~\ref{Cor:orderedExptimeHard}}
From Theorem~\ref{Thm:unorderedExptimeHard} we already know the
$\EXPTIME$-hardness of the
QCP for Boolean $\mDatalog(\tau_u)$-queries on unordered trees.
\\
Thus, it suffices to give a polynomial-time reduction from this
problem to
the QCP for Boolean $\mDatalog(\tau_o)$-queries on ordered trees.

For this, note that $\tau_u\subseteq \tau_o^{\Child}$.
Thus, upon input of two Boolean $\mDatalog(\tau_u)$-queries $Q_1$ and
$Q_2$, we can apply Corollary~\ref{remark_1} to compute, in linear time,
two Boolean $\mDatalog(\tau_o)$-queries $Q'_1$ and $Q'_2$ such that
$\AF{Q'_i}(T)=\AF{Q_i}(T)$ is true for all \emph{ordered} trees $T$ and each $i\in\set{1,2}$.
Furthermore, since $Q_i$ is of schema $\tau_u$, we have that
$\AF{Q_i}(T)=\AF{Q_i}(\tilde{T})$ is true for all
ordered trees $T$ and their unordered version $\tilde{T}$.
Thus, we have $Q_1\subseteq Q_2$ iff $Q'_1\subseteq Q'_2$. 
I.e., we have established a polynomial-time reduction from the QCP for
unordered trees to the QCP for ordered trees. This completes the proof
of Corollary~\ref{Cor:orderedExptimeHard}.
\end{proofof}

\clearpage
\section{$\EXPTIME$-Membership: Proof of Theorem~\ref{Thm:orderedInExptime}}\label{appendix:orderedInExptime}

The aim of this appendix is to prove the following Theorem:

\setcounter{restoreAppTheorem}{\value{theorem}}
\setcounter{theorem}{\value{Counter_Thm:orderedInExptime}}
\begin{Theorem}  \textbf{(restated)}
 The QCP for unary
 $\mDatalog(\tauGK^{\Child})$-queries on ordered trees belongs to $\EXPTIME$.
\end{Theorem}

\setcounter{theorem}{\value{restoreAppTheorem}}

We proceed as described in the proof sketch given in
Section~\ref{section:QCPofMonadicDatalog}.

\subsection{Step~\eqref{item:Step1}: From unary queries to Boolean queries}

Let $\Sigma$ be a finite alphabet, let $T$ be an ordered $\Sigma$-labeled
tree, and let $v$ be a node of $T$.
Considering the extended alphabet $\Sigma'\deff\Sigma\times\set{0,1}$,
we represent the tuple $(T,v)$ by an ordered $\Sigma'$-labeled tree
$T'_v$ as follows: $T'_v$ is obtained from $T$ by changing the node
labels, so that node $v$ receives label $(\alpha_v,1)$, and all
further nodes $u$ receive label $(\alpha_u,0)$, where $\alpha_v$ and
$\alpha_u$ denote the nodes' labels in $T$.

\begin{Lemma}\label{lemma:Unary2BooleanQueries}
Every unary $\mDatalog(\tauGKAlph{\Sigma}^{\Child})$-query $Q$
can be rewritten, in linear time, into a Boolean
$\mDatalog(\tauGKAlph{\Sigma'}^{\Child})$-query $\Bool{Q'}$
which satisfies the following:
\begin{mi}
 \item
   For every ordered $\Sigma$-labeled tree $T$ and every node $v$ of
   $T$ we have\\ $v\in\AF{Q}(T) \iff \AF{\Bool{Q'}}(T'_v)=\Yes$.
 \smallskip
 \item
   For every ordered $\Sigma'$-labeled tree $T'$ with
   $\AF{\Bool{Q'}}(T')=\Yes$, there are an ordered $\Sigma$-labeled tree
   $T$ and a node $v$ of $T$ such that $T'=T'_v$.
\end{mi}
\end{Lemma}
\begin{proof}
 Let $Q=(\PP,P)$. We will construct $\Bool{Q'}$ as follows:
 \begin{mei}
   \item\label{item:i-Unary2Bool}
     $\Bool{Q'}$ will simulate the program $\PP$ of $Q$. 
   \item\label{item:ii-Unary2Bool}
     In parallel, $\Bool{Q'}$ checks that the input tree contains
     exactly one node whose label is of the form $(\alpha,1)$ for some
     $\alpha\in\Sigma$.
     We construct $\Bool{Q'}$ in such a way that this is true iff the
     input tree's root node receives the intensional predicate $C_1$.
   \item\label{item:iii-Unary2Bool}
     Finally, the root node receives the query predicate of
     $\Bool{Q'}$ iff it has the $C_1$-predicate \emph{and} 
     the query predicate $P$ of the query $Q$ contains
     a node of label $(\alpha,1)$, for some $\alpha\in \Sigma$.
 \end{mei}
 To this end, we let $\Bool{Q'}$ be specified by a monadic datalog
 program $\PP'$ and a query predicate $P'$ chosen as follows:

 Start with $\PP'\deff\emptyset$.
 For each letter $\alpha\in\Sigma$, we add to $\PP'$ the rules
 \begin{align*}
    \Label_\alpha(x) & \leftarrow \Label_{(\alpha,0)}(x) 
   & X_0(x)&\leftarrow \Label_{(\alpha,0)}(x)
\\
    \Label_\alpha(x) & \leftarrow \Label_{(\alpha,1)}(x) 
   & X_1(x)&\leftarrow \Label_{(\alpha,1)}(x)
 \end{align*}
 where $X_0$ and $X_1$ are unary relation symbols that do not occur in $\PP$.

 Next, add to $\PP'$ all rules of $\PP$. Note that this way, we ensure
 that $\PP'$ simulates $\PP$, and hence \eqref{item:i-Unary2Bool} is
 achieved.

 To achieve \eqref{item:ii-Unary2Bool}, we use two
 intensional predicates $C_0,C_1$. We choose rules that proceed
 the binary tree built by the $\Fc$ and $\Ns$ relations in a bottom-up
 manner and propagates, via the predicates $C_0$ and $C_1$, whether the 
 subtree rooted at the current node contains exactly $0$ or exactly $1$
 nodes that carry the predicate $X_1$. This is achieved by the
 following list of rules, which we add to $\PP'$:
 \begin{align*}
   C_0(x)&\leftarrow \Leaf(x), \Ls(x), X_0(x)
 \\
   C_1(x)&\leftarrow \Leaf(x), \Ls(x), X_1(x)
 \\[2ex] 
   C_0(x)&\leftarrow \Leaf(x), \Ns(x,y), X_0(x), C_0(y)
\\
   C_1(x)&\leftarrow \Leaf(x), \Ns(x,y), X_0(x), C_1(y)
\\
   C_1(x)&\leftarrow \Leaf(x), \Ns(x,y), X_1(x), C_0(y)
 \\[2ex] 
   C_0(x)&\leftarrow \Ls(x), \Fc(x,y), X_0(x), C_0(y)
 \\
   C_1(x)&\leftarrow \Ls(x), \Fc(x,y), X_0(x), C_1(y)
 \\
   C_1(x)&\leftarrow \Ls(x), \Fc(x,y), X_1(x), C_0(y)
 \\[2ex]
   C_0(x)&\leftarrow \Fc(x,y), \Ns(x,z), X_0(x), C_0(y), C_0(z)
 \\
   C_1(x)&\leftarrow \Fc(x,y), \Ns(x,z), X_0(x), C_0(y), C_1(z)
 \\
   C_1(x)&\leftarrow \Fc(x,y), \Ns(x,z), X_0(x), C_1(y), C_0(z)
 \\
   C_1(x)&\leftarrow \Fc(x,y), \Ns(x,z), X_1(x), C_0(y), C_0(z)
 \end{align*}

Finally, we achieve \eqref{item:iii-Unary2Bool} by letting $P'$ be a
new intensional predicate and by adding to $\PP'$
the rule
\begin{align*}
  P'(x) & \leftarrow \Root(x), C_1(x), P(y), X_1(y).
\end{align*}
Clearly, $\PP'$ can be generated in time linear in the size of $Q$.
\end{proof}

As an immediate consequence, we obtain:

\begin{Lemma}\label{lemma:Unary2BooleanQCP}
 Let $\Sigma$ be a finite alphabet and let
 $\Sigma'\deff\Sigma\times\set{0,1}$. 
 Within linear time, we can rewrite given unary
 $\mDatalog(\tauGKAlph{\Sigma}^{\Child})$-queries $Q_1$ and $Q_2$ into
 Boolean $\mDatalog(\tauGKAlph{\Sigma'}^{\Child})$-queries $Q'_1$ and
 $Q'_2$ such that $Q_1\subseteq Q_2$ iff $Q'_1\subseteq Q'_2$.
\end{Lemma}
\begin{proof}
 For each $i\in\set{1,2}$ let $Q'_i$ be the query obtained by
 Lemma~\ref{lemma:Unary2BooleanQueries}.

 In case that $Q_1\not\subseteq Q_2$, there are an ordered
 $\Sigma$-labeled tree $T$ and a node $v$ of $T$ such that $v\in
 \AF{Q_1}(T)$ and $v\not\in\AF{Q_2}(T)$. By
 Lemma~\ref{lemma:Unary2BooleanQueries} we obtain that
 $\AF{Q'_1}(T'_v)=\Yes$ and $\AF{Q'_2}(T'_v)=\No$. 
 Thus, $Q'_1\not\subseteq Q'_2$.

 In case that $Q'_1\not\subseteq Q'_2$, there is an ordered
 $\Sigma'$-labeled tree $T'$ such that $\AF{Q'_1}(T')=\Yes$ and
 $\AF{Q'_2}(T')=\No$.
 Since $\AF{Q'_1}(T')=\Yes$, Lemma~\ref{lemma:Unary2BooleanQueries}
 tells us that there are an ordered $\Sigma$-labeled tree $T$ and a
 node $v$ of $T$ such that $T'=T'_v$. Furthermore, by
 Lemma~\ref{lemma:Unary2BooleanQueries} we know that $v\in
 \AF{Q_1}(T)$ and $v\not\in\AF{Q_2}(T)$. Thus, $Q_1\not\subseteq Q_2$.
\end{proof}

Finally, we use Theorem~\ref{Theorem:GK_TMNF} to eliminate the
$\Child$-predicate and to obtain queries in TMNF.

\begin{Proposition}\label{prop:Step1}
 Let $\Sigma$ be a finite alphabet and let
 $\Sigma'\deff\Sigma\times\set{0,1}$. 
 Within linear time, we can rewrite given
 unary $\mDatalog(\tauGKAlph{\Sigma}^{\Child})$-queries $Q_1$ and $Q_2$ into
 Boolean $\mDatalog(\tauGKAlph{\Sigma'})$-queries $Q'_1$ and
 $Q'_2$ such that $Q_1\subseteq Q_2$ iff $Q'_1\subseteq Q'_2$.
 Furthermore, the programs of $Q'_1$ and $Q'_2$ are in TMNF.
\end{Proposition}
\begin{proof}
 We apply Lemma~\ref{lemma:Unary2BooleanQCP} to obtain Boolean
 queries $Q'_1$ and $Q'_2$. Afterwards, we
 apply Theorem~\ref{Theorem:GK_TMNF} to rewrite the programs of the
 queries $Q'_1$ and $Q'_2$ into programs in TMNF of schema $\tauGKAlph{\Sigma'}$.
\end{proof}

Note that Proposition~\ref{prop:Step1} partially establishes
Step~\eqref{item:Step1} of the agenda described in
Section~\ref{section:QCPofMonadicDatalog}.

\subsection{Step~\eqref{item:Step1}: From Ordered Unranked Trees to Binary Trees}

For achieving Steps \eqref{item:Step2} and \eqref{item:Step3} we use,
among other things, the classical notion of nondeterministic tree
automata, which operate on ordered \emph{binary} $\Sigma$-labeled
trees. This subsection's goal is to fix notations concerning binary
trees, and to show that, in order to prove
Theorem~\ref{Thm:orderedInExptime}, it suffices to find a 1-fold exponential
algorithm that solves the QCP for Boolean queries in TMNF regarding binary trees.

\subsubsection{Binary trees.}

An \emph{ordered $\Sigma$-labeled binary tree} (for short: binary
tree) $T=(V^T,\lambda^T,L^T,R^T)$ consists of a finite set
$V^T$ of nodes, a function $\lambda^T:V^T\to\Sigma$ assigning to each
node $v$ of $T$ a label $\lambda^T(v)\in\Sigma$, and disjoint sets
$L^T,R^T\subseteq V^T\times V^T$ such that the graph $(V^T,E^T)$
with $E^T\deff L^T\cup R^T$ is a rooted directed tree where edges are
directed from the root to the leaves, and each node
has at most 2 children. For a tuple $(u,v)\in L^T$ (resp.,
$R^T$), we say that node $v$ is the \emph{left child} (resp., the
\emph{right child}) of node $u$. 

We represent such a tree $T$ as a relational structure of domain $V^T$
with unary and binary relations: For each label $\alpha\in\Sigma$,
$\Label_\alpha(x)$ expresses that $x$ is a node with label $\alpha$;
$\Lc(x,y)$ (resp., $\Rc(x,y)$) expresses that $y$ is the left (resp.,
right) child of node $x$;
$\Root(x)$ expresses that $x$ is the tree's root node; 
$\Hnlc(x)$ (resp., $\Hnrc(x)$) expresses that node $x$ has no left
child (resp., no right child), i.e., there is no node $y$ with
$(x,y)\in L^T$ (resp., $R^T$).

We denote this relational structure representing $T$ by $\S_b(T)$, but
when no confusion arises we simply write $T$ instead of $\S_b(T)$.
This relational structure is of schema
\[
  \taubAlph{\Sigma} \ \deff \ 
  \set{\Lc,\Rc}\cup\set{\Root,\Hnlc,\Hnrc}\cup\setc{\Label_\alpha}{\alpha\in\Sigma}.
\]

\subsubsection{Representing Ordered Unranked Tress by Binary Trees.} 

We use (a variant of) the standard representation (cf., e.g.,
\cite{DBLP:conf/csl/Neven02}) of ordered unranked trees by 
binary trees.
We represent an ordered $\Sigma$-labeled (unranked) tree $T$ by a
binary tree $\Bin{T}$ as follows:
$\Bin{T}$ has the same vertex set and the same node labels as $T$, 
the ``left child'' relation $L^{\Bin{T}}$ consists of all tuples $(x,y)$
such that $y$ is the first child of $x$ in $T$ (i.e., $\Fc(x,y)$ is
true in $\S_o(T)$), and the ``right child'' relation $R^{\Bin{T}}$
consists of all tuples $(x,y)$ such that $y$ is the next sibling of
$x$ in $T$ (i.e., $\Ns(x,y)$ is true in $\S_o(T)$).

Note that the relational structure $\S_b(\Bin{T})$ is obtained from
the structure $\S_o(T)$ as follows:
\begin{mi}
 \item drop the relations $\Child$ and $\Desc$,
 \item rename the relations $\Fc$, $\Ns$, $\Leaf$, $\Ls$ into
    $\Lc$, $\Rc$, $\Hnlc$, $\Hnrc$, and
 \item
   insert the root node into the relation $\Hnrc$.
\end{mi}
Furthermore, note that for a binary tree $T'$ there exists an unranked
ordered tree $T$ with $T'=\Bin{T}$ if, and only if, the root of $T'$
has no right child (and in this case the tree $T$ is unique).

\begin{lemma}\label{lemma:unranked2binary}
  Every Boolean $\mDatalog(\tauGKAlph{\Sigma})$-query $Q$ can be
  rewritten, in linear time, into a Boolean
  $\mDatalog(\taubAlph{\Sigma})$-query $Q'$ which satisfies the
  following:
  \begin{mi}
   \item
     For every ordered $\Sigma$-labeled (unranked) tree $T$ we have\\
     $\AF{Q}(T)=\Yes \iff \AF{Q'}(\Bin{T})=\Yes$.
   \item
     For every ordered $\Sigma$-labeled binary tree $T'$ with 
     $\AF{Q'}(T')=\Yes$ there is an ordered $\Sigma$-labeled
     (unranked) tree $T$ such that $T'=\Bin{T}$.
  \end{mi}
  Furthermore, if the program of $Q$ is in TMNF, then also the program
  of $Q'$ is in TMNF.
\end{lemma}
\begin{proof}
 Let $Q=(\PP,P)$.
 We specify $Q'$ by a monadic datalog program $\PP'$ and a query
 predicate $P'$ as follows:
 $\PP'$ is obtained from $\PP$ by renaming, in each rule, the
 predicates $\Fc$, $\Ns$, $\Leaf$, $\Ls$ into the predicates
 $\Lc$, $\Rc$, $\Hnlc$, $\Hnrc$.
 Furthermore, we let $P'$ be a new intensional predicate, and we add
 to $\PP'$ the rule
 \begin{align*}
   P'(x)&\leftarrow P(x),\Hnrc(x).
 \end{align*}
 It is straightforward to verify that the resulting Boolean query $Q'$
 has the desired properties.
\end{proof}

By combining this lemma with Proposition~\ref{prop:Step1}, we obtain
the following:

\begin{Proposition}\label{prop:BinaryTreesSufficeForQCP-upperBound}
 Let $\Sigma$ be a finite alphabet and let
 $\Sigma'\deff\Sigma\times\set{0,1}$. 
 Within linear time, we can rewrite given
 unary $\mDatalog(\tauGKAlph{\Sigma}^{\Child})$-queries $Q_1$ and
 $Q_2$ (querying ordered $\Sigma$-labeled unranked trees) into
 Boolean $\mDatalog(\taubAlph{\Sigma'})$-queries $Q'_1$ and
 $Q'_2$ (querying ordered $\Sigma'$-labeled binary trees) such that
 $Q_1\subseteq Q_2$ iff $Q'_1\subseteq Q'_2$. 
 Furthermore, the programs of $Q'_1$ and $Q'_2$ are in TMNF.
\end{Proposition}
\begin{proof}
 We first apply Proposition~\ref{prop:Step1} to obtain Boolean 
 $\mDatalog(\tauGKAlph{\Sigma'})$-queries $\tilde{Q}_1$ and
 $\tilde{Q}_2$, whose programs are in TMNF, such that $Q_1\subseteq Q_2$ iff
 $\tilde{Q}_1\subseteq\tilde{Q}_2$.

 Next, we apply Lemma~\ref{lemma:unranked2binary} to rewrite
 $\tilde{Q}_1$ and $\tilde{Q}_2$ into Boolean
 $\mDatalog(\taubAlph{\Sigma'})$-queries $Q'_1$ and $Q'_2$.
 It is straightforward to check that
 $\tilde{Q}_1\subseteq \tilde{Q}_2$ iff $Q'_1\subseteq Q'_2$:

 In case that $\tilde{Q}_1\not\subseteq\tilde{Q}_2$, there is an
 ordered $\Sigma'$-labeled unranked tree $T$ such that
 $\AF{\tilde{Q}_1}(T)=\Yes$ and $\AF{\tilde{Q}_2}(T)=\No$.
 By Lemma~\ref{lemma:unranked2binary} we obtain that 
 $\AF{Q'_1}(\Bin{T})=\Yes$ and 
 $\AF{Q'_2}(\Bin{T})=\No$. Thus, $Q'_1\not\subseteq Q'_2$.

 In case that $Q'_1\not\subseteq Q'_2$, there is an ordered
 $\Sigma'$-labeled binary tree $T'$ such that $\AF{Q'_1}(T')=\Yes$ and
 $\AF{Q'_2}(T')=\No$.
 Since $\AF{Q'_1}(T')=\Yes$, Lemma~\ref{lemma:unranked2binary} tells us
 that there is an ordered $\Sigma'$-labeled
 unranked tree $T$ such that $T'=\Bin{T}$. Furthermore, by
 Lemma~\ref{lemma:unranked2binary} we know that
 $\AF{\tilde{Q}_1}(T)=\Yes$ and $\AF{\tilde{Q}_2}(T)=\No$. Thus,
 $\tilde{Q}_1\not\subseteq \tilde{Q}_2$.
\end{proof}

Proposition~\ref{prop:BinaryTreesSufficeForQCP-upperBound} implies 
that, in order to prove Theorem~\ref{Thm:orderedInExptime}, it
suffices to show that the following problem can be solved in 1-fold
exponential time:

\begin{Problem}{Boolean-TMNF-QCP for monadic datalog on binary trees}
	\In A finite alphabet $\Sigma$ and two Boolean $\mDatalog(\taubAlph{\Sigma})$-queries $Q_1$ and
        $Q_2$ whose programs are in TMNF. 
        \Quest Is $Q_1\subseteq Q_2$\,?
\end{Problem}

\noindent
This finishes Step~\eqref{item:Step1} of the agenda described in
Section~\ref{section:QCPofMonadicDatalog}.

\subsection{Step~\eqref{item:Step2}: Nondeterministic Bottom-Up Tree
  Automata (NBTA)} \label{automaten}

In this subsection we recall the classical notion (cf., e.g.,
\cite{WThomas-Handbook-survey}) of nondeterministic 
bottom-up tree automata (NBTA, for short), and show that a Boolean
monadic datalog query $Q$ on binary trees can be translated, within
1-fold exponential time, into an NBTA $\Aut{A}_Q^{\No}$ which accepts
exactly those binary trees $T$ for which $\AF{Q}(T)=\No$.

\bigskip

A \emph{nondeterministice bottom-up tree automaton} (NBTA, for short)
$\Aut{A}$ is specified by a tuple $(\Sigma,S,\Delta,F)$, where
$\Sigma$ is a finite non-empty alphabet,
$S$ is a finite set of \emph{states}, $F\subseteq S$ is the set of
\emph{accepting states}, and $\Delta$ is the \emph{transition
  relation} with 
\begin{equation}\label{eq:Delta-NBTA}
  \Delta \quad  \subseteq \quad  S_{\#} \times S_{\#} \times \Sigma \times S,
\end{equation}
where $S_{\#}\deff S\cup\set{\#}$ for a symbol $\#$ that
does not belong to $S$.

A \emph{run} of $\Aut{A}$ on an ordered $\Sigma$-labeled binary tree
$T$ is a mapping $\rho:V^T\to S$ such that the following is true for
all nodes $v$ of $T$, where $\alpha$ denotes the label of $v$ in $T$:
\begin{mi}
 \item 
   If $v$ has no left child and no right child, then
   $\big(\#,\#,\alpha,\rho(v)\big)\in \Delta$.
 \item
   If $v$ has a left child $u_\ell$ and a right child $u_r$, then
   $\big(\rho(u_\ell),\rho(u_r),\alpha,\rho(v)\big)\in \Delta$.
 \item
   If $v$ has a left child $u_\ell$, but no right child, then
   $\big(\rho(u_\ell),\#,\alpha,\rho(v)\big)\in \Delta$.
 \item
   If $v$ has a right child $u_r$, but no left child, then
   $\big(\#,\rho(u_r),\alpha,\rho(v)\big)\in \Delta$.
\end{mi}
A run $\rho$ of $\Aut{A}$ on $T$ is \emph{accepting} if
$\rho(\rroot^T) \in F$, where $\rroot^T$ is the root node of $T$.
The automaton $\Aut{A}$ \emph{accepts} the tree $T$ if there exists an
accepting run of $\Aut{A}$ on $T$. A tree $T$ is \emph{rejected} iff
it is not accepted.
The \emph{tree language} $\lang(\Aut{A})$ is the set of all ordered
$\Sigma$-labeled binary trees $T$ that are accepted by $\Aut{A}$.
A set $L$ of ordered $\Sigma$-labeled binary trees is \emph{regular} if
$L=\lang(\Aut{A})$ for some NBTA $\Aut{A}$.

We define the \emph{size} $\size{\Aut{A}}$ of an NBTA $\Aut{A}$ to be the length of a
reasonable representation of the tuple $(\Sigma,S,\Delta,F)$; to be
precise, we let $\size{\Aut{A}}\deff |\Sigma| + |S| + |\Delta| + |F|$.
Note that due to \eqref{eq:Delta-NBTA} we have
\begin{equation}\label{eq:Size-NBTA}
 \size{A} \ = \ \ O(|S|^3{\cdot}|\Sigma|).
\end{equation}

It is well-known that the usual automata constructions for NFAs (i.e.,
nondeterministic finite automata on words) also apply to NBTAs. 
For formulating the results needed for our purposes, we introduce the
following notation: For finite alphabets $\Sigma$ and $\Gamma$ we let
$\Proj_\Sigma$ be the mapping from $\Sigma{\times}\Gamma$ to $\Sigma$
with $\Proj_\Sigma(\alpha,\beta)\deff \alpha$ for all
$(\alpha,\beta)\in \Sigma{\times}\Gamma$.
If $T$ is a $(\Sigma{\times}\Gamma)$-labeled tree, we write
$\Proj_\Sigma(T)$ to denote the $\Sigma$-labeled tree obtained from
$T$ by replacing each node label $(\alpha,\beta)$ by the node label $\alpha$.

By using standard automata constructions, one obtains:

\begin{Fact}[Folklore; see e.g.\ \cite{tata2008}]\label{fact:NBTA-constructions} \ 
\begin{description}
 \item[Union:] 
   For all NBTAs $\Aut{A}_1$ and $\Aut{A}_2$ over the same alphabet
   $\Sigma$, an NBTA $\Aut{A}_{\cup}$ with
   $\lang(\Aut{A}_\cup)=\lang(\Aut{A}_1)\cup\lang(\Aut{A}_2)$ can be
   constructed in time linear in $\size{\Aut{A}_1}$ and
   $\size{\Aut{A}_2}$. Furthermore, if $k_i$ is the number of states of
   $\Aut{A}_i$, for $i\in\set{1,2}$, then the number
   of states of $\Aut{A}_\cup$ is $k_1{+}k_2$.
 \item[Intersection:]
   For all NBTAs $\Aut{A}_1$ and $\Aut{A}_2$ over the same alphabet
   $\Sigma$, an NBTA $\Aut{A}_{\cap}$ with
   $\lang(\Aut{A}_\cap)=\lang(\Aut{A}_1)\cap\lang(\Aut{A}_2)$ can be
   constructed in time polynomial in $\size{\Aut{A}_1}$ and
   $\size{\Aut{A}_2}$. Furthermore, if $k_i$ is the number of states of
   $\Aut{A}_i$, for $i\in\set{1,2}$, then the number
   of states of $\Aut{A}_\cap$ is $k_1{\cdot}k_2$.
 \item[Complementation:]
   For every NBTA $\Aut{A}$, an NBTA $\Aut{A}^c$ which accepts exactly
   those trees that are rejected by $\Aut{A}$, can be constructed in
   time polynomial in $\size{\Aut{A}}{\cdot}2^k$, where $k$ denotes
   the number of states of $\Aut{A}$. 
   Furthermore, the number of states of $\Aut{A}^c$ is $2^k$.
 \item[Projection:]
   For every NBTA $\Aut{A}$ over an alphabet of the form
   $\Sigma{\times}\Gamma$, an NBTA $\Aut{A}^p$ over alphabet 
   $\Sigma$ with 
   \(
     \lang(\Aut{A}^p)
      =  
     \setc{\, \Proj_\Sigma(T)}{T \in \lang(\Aut{A}) \, }
   \)
   can be constructed in time polynomial in
   $\size{\Aut{A}}$.
   Furthermore, the number of states of $\Aut{A}^p$ is the same as the
   number of states of $\Aut{A}$.    
\end{description}
\end{Fact}

\medskip

\noindent
The emptiness problem for NBTAs is defined as follows:

\begin{Problem}{Emptiness problem for NBTAs}
  \In An NBTA $\Aut{A} =(\Sigma,S,\Delta,F)$.
  \Quest Is $\lang(\Aut{A}) = \emptyset$?
\end{Problem}
Similarly as for NFAs, the emptiness problem for NBTAs can be solved efficiently:

\begin{Fact}[Folklore; see e.g.\ \cite{tata2008}]\label{fact:emptiness}
 The emptiness problem for NBTAs can be solved in time polynomial in
 the size of the input automaton.
\end{Fact}

The following result establishes a relation between monadic datalog
and NBTAs.

\begin{Proposition} \label{query2notnbuta} \xdef\SecQzwei{\thesection}
  \xdef\SecQzweiDef{\theDef} 
Let $\Sigma$ be a finite alphabet and let $Q$ be a Boolean
$\mDatalog(\taubAlph{\Sigma})$-query whose program is in TMNF.
Within time polynomial in $|\Sigma|{\cdot} 2^{\size{Q}}$ we can construct an NBTA
$\Aut{A}^{\No}$ with $2^{O(\size{Q})}$ states, which accepts exactly those ordered
$\Sigma$-labeled binary trees $T$ where $\AF{Q}(T)=\No$.
\end{Proposition}
\begin{proof}
Our proof proceeds as described in the proof sketch given in
Section~\ref{section:QCPofMonadicDatalog}. Let $\PP$ be the
program of $Q$, let $X_1$ be the query predicate of $Q$, and let
$X_1,\ldots,X_n$ be the list of all intensional 
predicates of $\PP$.

\medskip

\noindent
\emph{Step~1: Transform $Q$ into an equivalent monadic
second-order sentence $\varphi_Q$:}
\\
We follow the ``standard construction'' (cf., 
\cite[Proposition~3.3]{GottlobKoch}), which uses the fact that the
result $\T_\PP^\omega(C)$ of a monadic datalog program $\PP$ on a set
$C$ of atomic facts is the \emph{least fixed-point} of the immediate
consequence operator $\T_\PP$ that contains $C$:

For any rule $r$ of $\PP$ of the form $h^r\leftarrow b_1^r, \ldots,b_m^r$,
define the formula
\[
  \psi_r \ \deff \ \
   \forall z_1\cdots \forall z_\ell \ \big(\; (b_1^r\und \cdots \und b_m^r) \impl
   h^r \; \big),
\]
where $z_1,\ldots,z_\ell$ is the list of variables appearing in the
rule $r$. Since $\PP$ is in TMNF, we know that $m=2$ and $\ell\leq
2$. W.l.o.g.\ we can assume that all rules use variables in $\set{z_1,z_2}$.

Let $\textit{SAT}(X_1,\ldots,X_n)$ be the conjunction of the formulas
$\psi_r$ for all rules $r$ in $\PP$, and let
\[
 \varphi_Q \ \deff \quad
 \forall X_1\cdots \forall X_n \ \big( \
 \textit{SAT}(X_1,\ldots,X_n)\impl X_1(\rroot)\ \big).
\]

It is staightforward to verify (see \cite[Proposition~3.3]{GottlobKoch}) that
for any ordered $\Sigma$-labeled binary tree $T$ we have
$\AF{Q}(T)=\Yes$ if, and only if, the tree $T$, expanded by a constant
$\rroot$ interpreted by the tree's root node, satisfies the 
$\MSO$-sentence $\varphi_Q$.

Clearly, $\varphi_Q$ is equivalent to
\ \(
 \forall X_1\cdots \forall X_n \; \big( \,  X_1(\rroot) \, \oder \, 
 \nicht\textit{SAT}(X_1,\ldots,X_n) \, \big)
\).
Furthermore, $\nicht\textit{SAT}$ is equivalent to
$\Oder_{r\in\PP}\nicht\psi_r$; and $\nicht\psi_r$ is equivalent to the formula
$\exists z_1\exists z_2\; (b_1^r\und b_2^r\und \nicht h^r)$, for a TMNF-rule $r$ of the
form $h^r\leftarrow b_1^r,b_2^r$.
In summary, we obtain that $\varphi_Q$ is equivalent to the formula
\[
 \varphi'_Q \ \deff \quad
 \forall X_1\cdots \forall X_n \ \exists z_1\exists z_2 \ \big( \
   X_1(\rroot) \ \oder \ 
   \Oder_{r\in\PP} \big( b_1^r \und b_2^r \und \nicht h^r \big)
 \ \big).
\]

Clearly, for any tree $T$ we have $\AF{Q}(T)=\No$ iff $T$ satisfies
the formula $\nicht\varphi'_Q$, which is equivalent to the formula
\[
 \tilde{\varphi}_Q \ \deff \quad
 \exists X_1\cdots \exists X_n \ \nicht\ \exists z_1\exists z_2 \ \big( \
   X_1(\rroot) \ \oder \ 
   \Oder_{r\in\PP} \big( b_1^r \und b_2^r \und \nicht h^r \big)
 \ \big).
\]

\bigskip

\noindent
\emph{Step~2: Transform $\tilde{\varphi}_Q$ into an equivalent
  NBTA:}\\
We proceed in the same way as in well-known textbook proofs for
B\"uchi's Theorem, resp., the Theorem by Doner and Thatcher and Wright
(stating the equivalence of $\MSO$-definable languages and regular
languages (of finite words and trees, respectively); cf.\ e.g.\ \cite{WThomas-Handbook-survey,FlumGrohe-ParamKompl}):

Based on the formula $\tilde{\varphi}_Q$ we give the construction of
the desired NBTA $\Aut{A}^{\No}$
along the composition of the formula. 

For the induction base, we have to handle quantifier-free formulas occurring in
$\tilde{\varphi}_Q$.
For this, we consider trees over alphabet
$\Sigma_n\deff\Sigma\times\Gamma\times\Gamma'$ for $\Gamma\deff\set{0,1}^n$ and
$\Gamma'\deff\set{0,1}^2$. 
If a node $v$ has label $(\alpha,\gamma,\gamma')$, for
$\gamma=\gamma_1\cdots\gamma_n$ and $\gamma'=\gamma'_1\gamma'_2$, we interpret this
as the information that $v$ has $\Sigma$-label $\alpha$, belongs to
the relation $X_i$ iff $\gamma_i=1$, and is the value of the variable
$z_j$ iff $\gamma'_j=1$ (for $i\in\set{1,\ldots,n}$ and $j\in\set{1,2}$).
We will refer to $\gamma'_j$ (resp., $\gamma_i$ and $\alpha$) as the
$z_j$-component (resp., the $X_i$-component and the $\Sigma$-component) of the label. 

To check that the values in the $z_j$-components of a labeling
indeed represent a variable assignment, we build for each
$j\in\set{1,2}$ an NBTA $\Aut{A}_{z_j}$ that accepts exactly those
$\Sigma_n$-labeled trees where exactly one
node carries a label whose $z_j$-component is 1. 
For example, the NBTA
$\Aut{A}_{z_2}$ can be chosen as $(\Sigma_n,S,\Delta,F)$ with 
$S=\set{s_0,s_1}$, $F=\set{s_1}$, and $\Delta$ consisting of the
transitions
\begin{center}
 $(\#,\#,\beta,s_\nu)$, \ 
 $(s_0,s_0,\beta,s_\nu)$, \
 $(s_0,\#,\beta,s_\nu)$, \
 $(\#,s_0,\beta,s_\nu)$
\end{center}
for all $\nu\in\set{0,1}$ and all labels
$\beta\in\Sigma\times\Gamma\times\set{0,1}\times\set{\nu}$, \ and the transitions
\begin{center}
 $(s_1,\#,\beta,s_1)$, \
 $(\#,s_1,\beta,s_1)$, \
 $(s_1,s_0,\beta,s_1)$, \
 $(s_0,s_1,\beta,s_1)$
\end{center}
for all labels $\beta\in\Sigma\times\Gamma\times\set{0,1}\times\set{0}$.
This automaton performs a bottom-up scan of the tree and remains in state
$s_0$ until it encounters a node whose label has a 1 in its
$z_2$-component. The latter induces a change into state $s_1$. The
automaton gets stuck (i.e., no run exists) if it is in state
$s_1$ and encounters another node whose label has a 1 in its $z_2$-component.  

To check whether an atomic or negated atomic formula $\chi$ (occurring in $\tilde{\varphi}_Q$)
is satisfied by an input
tree, we build an NBTA $\Aut{A}_\chi$ that accepts an input tree $T$ iff $T$
contains, for each variable $z_j$ occurring in $\chi$, a node $v_j$ whose $z_j$-component
is 1, such that the nodes $v_j$ satisfy $\chi$.
If $\chi$ involves a \emph{unary} atom, this can be achieved in a
straightforward way using an automaton with 2 states.
If $\chi$ is a \emph{binary} atom, this is not difficult either.
E.g., if $\chi=\Lc(z_2,z_1)$, the NBTA $\Aut{A}_{\chi}$ 
performs a bottom-up scan of the tree and remains in state $s_0$ until it
encounters a node $v_1$ whose $z_1$-component is labeled 1. The latter
induces a change into state $s_1$. From there on, the automaton either gets stuck, or 
it sees that $v_1$ is the left child of a node $v_2$ whose
$z_2$-component is one. The latter induces a change into an accepting
state $s_2$, which is propagated to the root.

Note that each of the NBTAs constructed so far has at most 3 states
and, according to \eqref{eq:Size-NBTA}, size $O(3^3{\cdot}|\Sigma_n|)= O(|\Sigma_n|)$.

The formula $\tilde{\varphi}_Q$ contains a conjunction $\zeta_r$ of the form
$(b_1^r\und b_2^r\und \nicht h^r)$, for each rule $r\in\PP$. 
We already have available NBTAs $\Aut{A}_{b_1^r}$, $\Aut{A}_{b_2^r}$,
$\Aut{A}_{\nicht h^r}$, $\Aut{A}_{z_1}$, $\Aut{A}_{z_2}$,
each of which has at most 3 states and size
$O(|\Sigma_n|)$. By using the intersection-construction
mentioned in Fact~\ref{fact:NBTA-constructions}, we can build the intersection automaton
$\Aut{A}_{\zeta_r}$ of these five NBTAs. This can be achieved in time
polynomial in $O(|\Sigma_n|)$; and the resulting automaton has
at most $3^5$ states and thus, due to \eqref{eq:Size-NBTA}, size $O(|\Sigma_n|)$.

The quantifier-free part of the formula $\tilde{\varphi}_Q$ is the
disjunction of the formula $X_1(\rroot)$ and the formulas $\zeta_r$, for
all $r\in \PP$. We already have available NBTAs
$\Aut{A}_{X_1(\rroot)}$ and $\Aut{A}_{\zeta_r}$ for each $r\in\PP$.
Using the union-construction mentioned in
Fact~\ref{fact:NBTA-constructions}, we can build the union automaton
$\Aut{A}_{\textit{qf}}$ of these automata. This can be achieved in time 
polynomial in $O(|\PP|{\cdot}|\Sigma_n|)$; and the resulting automaton has
at most $(|\PP|{+}1){\cdot}3^5 = O(|\PP|)$ states and thus, due to
\eqref{eq:Size-NBTA}, size $O(|\PP|^3{\cdot}|\Sigma_n|)$.

Note that $\Aut{A}_{\textit{qf}}$ is an NBTA over alphabet
$\Sigma{\times}\Gamma{\times}\Gamma'$.
We now use the projection-construction mentioned in
Fact~\ref{fact:NBTA-constructions} to build an NBTA $\Aut{A}_{\exists
  z_1\exists z_2}$ accepting the set of all trees of the form
$\Proj_{\Sigma\times\Gamma}(T)$, for $T$ accepted by
$\Aut{A}_{\textit{qf}}$.
The resulting automaton has the same number of states as
$\Aut{A}_{\textit{qf}}$, i.e., $O(|\PP|)$, has size 
$O(|\PP|^3{\cdot}|\Sigma{\times}\Gamma|)= O(|\PP|^3{\cdot}|\Sigma|{\cdot}2^n)$, and can be constructed
in time polynomial in $O(|\PP|^3{\cdot}|\Sigma_n|)$.

Next, we use the complementation-construction mentioned in
Fact~\ref{fact:NBTA-constructions} to build an NBTA
$\Aut{A}_{\nicht}$ which accepts exactly those
trees that are rejected by $\Aut{A}_{\exists z_1\exists z_2}$.
The automaton $\Aut{A}_{\nicht}$ has $2^{O(|\PP|)}$ states and thus, due to \eqref{eq:Size-NBTA}, size
$O(2^{O(|\PP|)}{\cdot}|\Sigma{\times}\Gamma|) =
O(2^{O(|\PP|)}{\cdot}|\Sigma|{\cdot}2^n)$.
It can be constructed in time polynomial in the size of
$\Aut{A}_{\exists z_1\exists z_2}$ and $2^{O(|\PP|)}$, i.e.,
polynomial in $2^{O(|\PP|)}{\cdot}|\PP|^3{\cdot}|\Sigma|{\cdot}2^n$

Finally, we use the projection-construction mentioned in
Fact~\ref{fact:NBTA-constructions} to build an NBTA $\Aut{A}^{\No}$
 accepting the set of all trees of the form
$\Proj_{\Sigma}(T)$, for $T$ accepted by $\Aut{A}_{\nicht}$. The
resulting automaton has the same number of states as
$\Aut{A}_{\nicht}$, i.e., $2^{O(|\PP|)}$ and can be constructed in
time polynomial in the size of $\Aut{A}_{\nicht}$, i.e., polynomial in
$2^{O(|\PP|)}{\cdot}|\Sigma|{\cdot}2^n =
|\Sigma|{\cdot}2^{n+O(|\PP|)} = |\Sigma|{\cdot}2^{O(\size{Q})}$.

It is straightforward to verify that the NBTA $\Aut{A}^{\No}$ accepts
exactly those $\Sigma$-labeled trees $T$ that satisfy the formula
$\tilde{\varphi}_Q$, i.e., those trees $T$ with $\AF{Q}(T)=\No$.
The entire construction of the automaton $\Aut{A}^{\No}$ took time
polynomial in $|\Sigma|{\cdot}2^{\size{Q}}$. This completes the proof
of Proposition~\ref{query2notnbuta}.
\end{proof}

This establishes the ``$\Aut{A}^{\No}$-part'' of 
Step~\eqref{item:Step2} of the agenda described in
Section~\ref{section:QCPofMonadicDatalog}.
By applying to $\Aut{A}^{\No}$ the complementation-construction
mentioned in Fact~\ref{fact:NBTA-constructions}, we obtain an NBTA
$\Aut{A}^{\Yes}$ which accepts exactly the $\Sigma$-labeled trees $T$
with $\AF{Q}(T)=\Yes$. However, the number of states of
$\Aut{A}^{\No}$ is $2^{O(\size{Q})}$, and hence the construction of 
$\Aut{A}^{\Yes}$ takes time polynomial in
$\size{\Aut{A}^{\No}}{\cdot}2^{2^{O(\size{Q})}}$, which is 2-fold
exponential in the size of the query $Q$.

To construct an NBTA equivalent to $\Aut{A}^{\Yes}$ within 1-fold
exponential time, we use a different automata model, described in the
next subsection.

\subsection{Step~\eqref{item:Step2}: 2-way alternating tree automata (2ATA)}

In this subsection we recall the notion (cf., e.g.,
\cite{CGKV,Vardi98,Maneth10TypeChecking}) of 2-way alternating tree 
automata (2ATA), and show that a Boolean monadic datalog query $Q$ on
binary trees can be translated, within polynomial time, into a 2ATA
$\ATAA^{\Yes}$ which accepts exactly those binary trees $T$ for
which $\AF{Q}(T)=\Yes$.
The following definitions concerning 2ATAs are basically taken
from \cite{CGKV,Vardi98}. 

\bigskip

For navigating in a binary tree $T$ we consider the 
operations $\Up,\Stay,\DownLeft,\DownRight$. They are viewed as
functions from $V^T_\bot$ to $V^T_\bot$ where
$V^T_\bot=V^T\cup\set{\bot}$ for the node set $V^T$ of
$T$ and a symbol $\bot$ not in $V^T$. Each of the operations in
$\Op\deff\set{\Up,\Stay,\DownLeft,\DownRight}$ maps $\bot$ to
$\bot$. Furthermore, for each node $v$ of $T$, we have 
$\Stay(v)=v$, while
$\Up(v)$ is the parent of $v$ in $T$ (resp.\ $\bot$, in case that
$v$ is the root of $T$), and
$\DownLeft(v)$ is the left child of $v$ in $T$ (resp.\ $\bot$, in
case that $v$ has no left child), and
$\DownRight(v)$ is the right child of $v$ in $T$ (resp.\ $\bot$, in
case that $v$ has no right child).

Let $M$ be a set. The set $\B^+(M)$ of \textit{positive Boolean
  formulas over $M$} contains 
all elements in $M$, 
and is closed under $\und$ and $\oder$.
For a set $M'\subseteq M$ and a formula $\theta\in \B^+(M)$, we say
that \emph{$M'$ satisfies $\theta$} iff 
assigning $\True$ to elements in $M'$ and $\False$ to elements in
$M\setminus M'$ makes $\theta$ true.

A \textit{two-way alternating tree automaton} (2ATA, for short)
$\ATAA$ is specified by a tuple $(\Sigma,S,s_0,\delta,F)$, where 
\begin{mi}
    \item $\Sigma$ is a finite non-empty alphabet,
    \item $S$ is a finite set of states,
    \item $s_0 \in S$ is the \emph{initial} state,
    \item $F \subseteq S$ is the set of \emph{accepting} states, and    
    \item $\delta: S \times \Sigma \to \B^+(S {\times} \Op)$ is the \emph{transition function}.
\end{mi}

\noindent 
As input, $\ATAA$ receives a $\Sigma$-labeled binary tree $T$.
It starts in the initial state $s_0$ at $T$'s root node.
Whenever $\ATAA$ is in a state $s\in S$ and currently visits a node $v$ of $T$ of label
$\alpha\in\Sigma$, it can either choose to stop its computation, or to
perform a further step in which the formula $\theta\deff \delta(s,\alpha)$
determines what is done next: the automaton nondeterministically
guesses a satisfying assignment for $\theta$, i.e., a set
$\set{\,(s_1,o_1),\ldots,(s_k,o_k)\,}$ (for some $k\geq 1$)
which satisfies $\theta$. Then, it starts $k$ independent copies of
$\ATAA$, namely a copy which
starts in state $s_i$ at node $o_i(v)$, for each
$i\in\set{1,\ldots,k}$.
In case that $o_i(v)=\bot$, the according automaton stops.
The acceptance condition demands that for every situation $(s,v)$ in
which the automaton stops, $s$ must be an accepting state.

This can be formalised by the following notion of a \emph{run} $R$, where 
the label $(s,o,v)$ of a node $w$ of $R$ denotes a transition into
state $s$ via the operation $o$ onto node $v$.

A \emph{run} of $\ATAA$ on a $\Sigma$-labeled binary tree $T$ is
a finite unordered \emph{unranked} $\Gamma$-labeled tree $R$, for
$\Gamma\deff S\times \Op \times V^T_\bot$,
which satisfies the following conditions:
\begin{me}
 \item 
   The root of $R$ is labeled with $(s_0,\Stay,\rroot^T)$, where $s_0$ is
   the initial state
   and $\rroot^T$ is the root of $T$.
 \item
   If $w$ is a node of $R$ that is labeled $(s,o,v)$ with $v=\bot$,
   then $w$ is a \emph{leaf} of $R$.
 \item 
   If $w$ is a node of $R$ that is labeled $(s,o,v)$ such that $v$
   is a node of $T$, and
   $w'$ is a child of $w$ in $R$ that is labeled $(s',o',v')$, then
   $v'=o'(v)$.
 \item 
   If $w$ is a node of $R$ that is labeled $(s,o,v)$ such that $v$
   is a node of $T$ labeled $\alpha\in\Sigma$, and $w$ has
   exactly $k$ children labeled
   $(s_1,o_1,v_1),\ldots,(s_k,o_k,v_k)$, then the formula
   $\theta\deff\delta(s,\alpha)$ is satisfied by the set
   $\set{\,(s_1,o_1),\ldots,(s_k,o_k)\,}$. 
\end{me} 

A run $R$ of $\ATAA$ on $T$ is \emph{accepting} if every leaf of $R$ is labeled with an
accepting state, i.e.: whenever $(s,o,v)$ is the label of a leaf of
$R$, we have $s\in F$.
The automaton $\ATAA$ \emph{accepts} the tree $T$ if there exists an
accepting run of $\ATAA$ on $T$.
The \emph{tree language} $\lang(\ATAA)$ is the set of all ordered
$\Sigma$-labeled binary trees $T$ that are accepted by $\ATAA$.

The \emph{size} $\size{\ATAA}$ of a 2ATA
$\ATAA$ is defined as the length of a reasonable repesentation of the tuple
$(\Sigma,S,S_0,\delta,F)$.  

It is known that 2ATAs accept exactly the same tree languages as
NBTAs, i.e., the \emph{regular} tree languages. Furthermore, there is
a 1-fold exponential algorithm that translates a 2ATA into an
equivalent NBTA:

\begin{Theorem}[Cosmadakis et al.\ \cite{CGKV}]\label{Thm:2ATAtoNBTA}\label{wata2nbuta}
 For every 2ATA $\ATAA$, an NBTA $\Aut{A}$ with
 $\lang(\Aut{A})=\lang(\ATAA)$ can be constructed within time 1-fold
 exponential in $\size{\ATAA}$.
\end{Theorem}

To be precise, \cite{CGKV} formulated the theorem not in terms of the
running time, but only in terms of the \emph{size} of the generated
NBTA. A proof sketch of the theorem can be found in \cite{CGKV}; detailed proofs of
more general results can be found in \cite{Vardi98,Maneth10TypeChecking}.

\bigskip

Our next goal is to find a polynomial-time algorithm which translates a Boolean
monadic datalog query $Q$ in TMNF into an equivalent 2ATA $\ATAA$
which accepts exactly those trees $T$ with $\AF{Q}(T)=\Yes$.

To construct such a 2ATA, we will exploit the
striking similarity between runs of 2ATAs and
\emph{proof trees} characterising the semantics of datalog (cf., the
textbook \cite{AHV}).
For constructing the desired 2ATA, the following observation
will be very convenient:

Let $Q$ be a Boolean $\mDatalog(\taubSigma)$-query whose program is in
TMNF, and let $\PP$ and $P$ be the program and the query predicate of $Q$.
For a $\Sigma$-labeled binary tree $T$ with root node $\rroot^T$ we have $\AF{Q}(T)=\Yes$ iff
there exists a proof tree $\PT$ for the fact $P(\rroot^T)$, such that the
leaves of the proof tree are labeled with facts in $\atoms(\S_b(T))$.
Note that, for the particular case of TMNF-programs, such a proof tree
$\PT$ has the following properties:

\begin{mi}
 \item
   The \emph{root} of $\PT$ is labeled with the atomic fact
   $P(\rroot^T)$.
 \item
   Each \emph{leaf} of $\PT$ is labeled with an atomic fact of one of the
   following forms:
   \begin{mi}
   \item
     $\Label_\alpha(v)$ where $\alpha\in\Sigma$ and $v$ is a node of
     $T$ labeled $\alpha$,
   \item 
     $\Root(\rroot^T)$, where $\rroot^T$ is the root of $T$,
   \item
     $\Hnlc(v)$ (resp., $\Hnrc(v)$), where $v$ is a node of $T$ that
     has no left child (resp., has no right child)
   \item
     $\Lc(v_1,v_2)$ (resp., $\Rc(v_1,v_2)$), where $v_2$ is the left
     (resp., right) child of $v_1$ in $T$.
   \end{mi}
 \item
   Each non-leaf node of $\PT$ is labeled with a fact
   $X(v)$ where $v$ is a node of $T$ and $X\in\idb(\PP)$.
 \item
   Every non-leaf node $w$ of $\PT$ has exactly 2 children $w_1$ and
   $w_2$. If $w$ is labeled by an atomic fact $X(v)$, then $\PP$
   contains a rule $r$ whose head is of the form $X(x)$, and the
   following is true:
   \begin{mea}
    \item 
     If the body of $r$ is of the form $Y(x),Z(x)$, then $w_1$ is
     labeled $Y(v)$ and $w_2$ is labeled $Z(v)$.
    \item
     If the body of $r$ is of the form $\Lc(x,y),Y(y)$ 
     then node $v$ of $T$ has a left child $v'$, and in
     $\PT$ the nodes $w_1$ and $w_2$ are labeled with the facts 
     $\Lc(v,v')$ and $Y(v')$.

     Accordingly, if the body of $r$ is of the form $\Rc(x,y),Y(y)$ 
     then node $v$ of $T$ has a right child $v'$, and in
     $\PT$ the nodes $w_1$ and $w_2$ are labeled with the facts 
     $\Rc(v,v')$ and $Y(v')$.

    \item
     If the body of $r$ is of the form $\Lc(y,x),Y(y)$,
     then node $v$ of $T$ is the left child of its parent $v'$, and in
     $\PT$ the nodes $w_1$ and $w_2$ are labeled with the facts 
     $\Lc(v',v)$ and $Y(v')$.

     Accordingly, if the body of $r$ is of the form $\Rc(y,x),Y(y)$,
     then node $v$ of $T$ is the right child of its parent $v'$, and in
     $\PT$ the nodes $w_1$ and $w_2$ are labeled with the facts 
     $\Rc(v',v)$ and $Y(v')$.
   \end{mea}
\end{mi}
We will build a 2ATA for which an accepting run $R$ on an input tree $T$
precisely corresponds to a proof tree $\PT$ for the fact
$P(\rroot^T)$.
To better cope with technical details in the automaton construction,
we will consider automata which receive input trees that are labeled by
the extended alphabet $\hat{\Sigma}$, with
\[
  \hat{\Sigma} \ \deff \ \
  \Sigma \times 2^{\set{\,\Root,\ \Hnlc,\ \Hnrc,\ \Islc,\ \Isrc\,}}.
\]
With every $\Sigma$-labeled binary tree $T$ we associate a
$\hat{\Sigma}$-labeled binary tree $\hat{T}$ that is obtained from
$T$ by replacing the label of each node $v$ labeled $\alpha\in\Sigma$
with the label $(\alpha,I)$ where
$I\subseteq\set{\,\Root,\ \Hnlc,\ \Hnrc,\ \Islc,\ \Isrc\,}$ is given as follows:
\begin{eqnarray*}
   \Root\in I &\iff & \text{$v$ is the root of $T$},
\\
   \Hnlc\in I &\iff & \text{$v$ is a node of $T$ that has no left
     child},
\\
   \Hnrc\in I &\iff & \text{$v$ is a node of $T$ that has no right child},
\\
   \Islc\in I & \iff & \text{$v$ is the left child of its parent $v'$
     in $T$},
\\
   \Isrc\in I &\iff & \text{$v$ is the right child of its parent $v'$ in $T$}.
\end{eqnarray*}

We are now ready for this subsection's key result:

\begin{Proposition} \label{query22wata}  \xdef\queryWata{\thesection} \xdef\queryWataDef{\theDef} 
Let $\Sigma$ be a finite alphabet and let $Q$ be a Boolean
$\mDatalog(\taubSigma)$-query whose program is in TMNF.
Within time polynomial in the size of $Q$ and $\Sigma$, we can
construct a 2ATA $\ATAA$ such that for all
$\Sigma$-labeled binary trees $T$, the automaton $\ATAA$
accepts the tree $\hat{T}$ if, and only if, $\AF{Q}(T)=\Yes$.
\end{Proposition}
\begin{proof}
Let $\PP$ and $P$ be the program and the query predicate of $Q$.
We construct the automaton $\ATAA$ in such a way that a proof
tree $\PT$ for the fact $P(\rroot^T)$ can be easily be turned into an accepting
run of $\ATAA$ on $\hat{T}$ (and vice versa).

The state set $S$ of the
$\ATAA=(\hat{\Sigma},S,s_0,\delta,F)$ is chosen as the set all 
intensional predicates of $\PP$, all
unary relation symbols in $\taubSigma$, and additionally, we use
states called $\Islc$, $\Isrc$, $\Accept$, and $\Reject$. I.e.,
\[
\begin{array}{ll}
  S \ = \ \ 
&
  \set{\Accept,\ \Reject} 
  \ \cup \ 
  \idb(\PP) 
  \ \cup \
\\[1ex] 
&
  \setc{\Label_\alpha}{\alpha\in\Sigma} 
  \ \cup \  
  \set{\,\Root,\ \Hnlc,\ \Hnrc,\ \Islc,\ \Isrc\,}.
\end{array}
\]
The query predicate $P$ is the initial state, and $\Accept$ is the
only accepting state. I.e., $s_0\deff P$ and $F\deff\set{\Accept}$.
\\
The transition function
$\delta:S\times\hat{\Sigma}\to\B^+(S{\times}\Op)$
is chosen as follows:
\\
Let $\beta=(\alpha,I)$ be an arbitrary letter in $\hat{\Sigma}$. 
We let
\[
  \delta(\Accept,\beta) \deff (\Accept,\Stay)
  \qquad \text{und}\qquad
  \delta(\Reject,\beta) \deff (\Reject,\Stay).
\]
For every $\alpha'\in\Sigma$ we let
\[
   \delta(\Label_{\alpha'},\beta) \deff \left\{
     \begin{array}{cl}
       (\Accept,\Stay) & \text{ if } \alpha'=\alpha \\[1ex]
       (\Reject,\Stay) & \text{ otherwise}.
     \end{array}
   \right.
\]
For every $X\in\set{\,\Root,\ \Hnlc,\ \Hnrc,\ \Islc,\ \Isrc\,}$ we let
\[
  \delta(X,\beta) \deff \left\{
    \begin{array}{cl}
       (\Accept,\Stay) & \text{ if } X\in I \\[1ex]
       (\Reject,\Stay) & \text{ otherwise}.
    \end{array}
  \right.
\]

\noindent
For the case that $X\in\idb(\PP)$, the formula $\delta(X,\beta)$ is
specified as follows.
We let $\PP_X$ be the set of all rules of $\PP$ whose head is of the form $X(x)$, and we choose
\[
  \delta(X,\beta) \deff \Oder_{r\in\PP_X} \theta_r,
\]
where the formula
$\theta_r\in\B^+(S{\times}\Op)$ is
chosen as indicated in the following table:

\begin{center}
\begin{tabular}{lll}
 rule $r$ of the form & \hspace{5mm} & conjunction $\theta_r$
 \\[0.5ex] \hline \\[-1ex]
 $X(x)\leftarrow Y(x),Z(x)$     & & $(Y,\Stay)\;\und\; (Z,\Stay)$ \\[1ex]
 $X(x)\leftarrow \Lc(x,y),Y(y)$ & & $(\Islc,\DownLeft) \;\und\; (Y,\DownLeft)$ \\[1ex]
 $X(x)\leftarrow \Rc(x,y),Y(y)$ & & $(\Isrc,\DownRight) \;\und\; (Y,\DownRight)$ \\[1ex]
 $X(x)\leftarrow \Lc(y,x),Y(y)$ & & $(\Islc,\Stay)\;\und\;(Y,\Up)$ \\[1ex]
 $X(x)\leftarrow \Rc(y,x),Y(y)$ & & $(\Isrc,\Stay)\;\und\;(Y,\Up)$ \\[1ex]
\end{tabular}
\end{center}

\medskip

Clearly, this automaton $\ATAA$ can be constructed in time polynomial
in the size of $\Sigma$ and $Q$.
It remains to verify that, indeed, for any $\Sigma$-labeled binary
tree $T$ we have $\AF{Q}(T)=\Yes$ $\iff$ $\ATAA$ accepts $\hat{T}$.

\medskip

For the ``$\Longrightarrow$''-direction, let $\PT$ be a proof tree for the
the fact $P(\rroot^T)$. We can transform $\PT$ into a run $R$ of
$\ATAA$ on $\hat{T}$ as follows: Assign the
new label $(P,\Stay,\rroot^T)$ to the root node of $\PT$.
For each non-leaf node $w$ of $\PT$ note that $w$ is originally
labeled by an atomic fact $X(v)$ with
$X\in\idb(\PP)$, and $w$ has exactly two children $w_1,w_2$ in $\PT$.
\begin{mea}
 \item 
   If $w_1,w_2$ are labeled $Y(v),Z(v)$, then assign to node $w_1$ the
   new label $(Y,\Stay,v)$ and to node $w_2$ the new label
   $(Z,\Stay,v)$.
 \item
   If $w_1,w_2$ are labeled $\Lc(v,v'),Y(v')$, then 
   assign to node $w_1$ the new label $(\Islc,\DownLeft,v')$ and to node
   $w_2$ the new label $(Y,\DownLeft,v')$.
   Furthermore, we add to $w_1$ a new child labeled $(\Accept,\Stay,v)$.
   
   We proceed analogously in case that $w_1,w_2$ is labeled $\Rc(v,v'),Y(v')$. 
 \item
   If $w_1,w_2$ are labeled $\Lc(v',v),Y(v')$, then 
   assign to $w_1$ the new label
   $(\Islc,\Stay,v)$, and to node $w_2$ the new label
   $(Y,\Up,v')$. 
   Furthermore, we add to $w_1$ a new child labeled $(\Accept,\Stay,v)$.
   
   We proceed analogously in case that $w_1,w_2$ is labeled $\Rc(v',v),Y(v')$. 
\end{mea}
Finally, for each leaf $w$ of $\PT$ that was originally labeled
$X(v)$ for an $X\in\set{\Root,\Hnlc,\Hnrc} \cup
\setc{\Label_\alpha}{\alpha\in\Sigma}$, we add a new child $w_1$ 
that receives the new label $(\Accept,\Stay,v)$.

It is straightforward to verify 
that the obtained tree $R$ is an
accepting run of $\ATAA$ on $\hat{T}$.

\medskip

For the direction ``$\Longleftarrow$'' let $R$ be an accepting run of
$\ATAA$ on $\hat{T}$. Along the definition of $\delta$ it is
straightforward to see that we can assume w.l.o.g.\  that each node of
$R$ has at most 2 children.

The run $R$ can be turned into a proof tree
$\PT$ for the fact $P(\rroot^T)$ (i.e., 
witnessing that $\AF{Q}(T)=\Yes$) as follows:
Consider each node $w$ of $R$, and let $(s,o,v)$ be the label of node
$w$.

Since $R$ is an \emph{accepting} run and $\Accept$ is the only
accepting state, we know by the construction of $\delta$ that $s\neq \Reject$, and that $v\neq\bot$ if
$s\neq\Accept$. 
In case that $s\in\taubSigma\cup\idb(\PP)$, we assign to $w$ the new
label ``$s(v)$''.

In case that $s=\Label_{\alpha'}$ for an $\alpha'\in\Sigma$, we know by
the construction of $\delta$ and 
the fact that $R$ is an \emph{accepting} run, that node $w$ has a
unique child $w_1$ in $R$, and this node $w_1$ is labeled
with $(\Accept,\Stay,v)$.
Furthermore, we know by the construction of $\delta$ that $\alpha'=\alpha$
where $\beta=(\alpha,I)$ is the label of node $v$ in $\hat{T}$. Thus, the
statement ``$\Label_{\alpha'}(v)$'' is true for node $v$ in $T$.
Hence, we delete the node $w_1$ (and all nodes in the subtree
rooted at $w_1$).

In case that $s\in\set{\Root,\Hnlc,\Hnrc,\Islc,\Isrc}$, we know by the construction of $\delta$ and
the fact that $R$ is an \emph{accepting} run, that node $w$ has a
unique child $w_1$ in $R$, and this node $w_1$ is labeled
with $(\Accept,\Stay,v)$.
Furthermore, we know by the construction of $\delta$ that $s\in I$,
where $\beta=(\alpha,I)$ is the label of node $v$ in $\hat{T}$. Thus, the
statement ``$s(v)$'' is true for node $v$ in $T$.
Hence, we delete the node $w_1$ (and all nodes in the subtree
rooted at $w_1$).
\\
In case that $s\in\set{\Root,\Hnlc,\Hnrc}$, the node $w$ then is a
leaf, labeled with an atomic fact ``$s(v)$'' that is true in $T$.
\\
In case that $s=\Islc$, the statement ``$\Islc(v)$'' is a true
statement, but it is not suitable as label in a proof tree, since the
predicate $\Islc$ does not belong to the schema $\taubSigma$. 
Therefore, we replace the label ``$\Islc(v)$'' by the label
``$\Lc(v',v)$'' where $v'$ is the parent of $v$ in $T$.
We proceed analogously in case that $s=\Isrc$.

It is straightforward to verify 
that the obtained tree $\PT$ is a proof tree for $P(\rroot^T)$.
This completes the proof of Proposition~\ref{query22wata}.
\end{proof}

Finally, we are ready for establishing the second part of 
Step~\ref{item:Step2} of the agenda described in
Section~\ref{section:QCPofMonadicDatalog}.

\begin{Proposition}\label{prop:mDatalog2Ayes}
 Let $\Sigma$ be a finite alphabet and let $Q$ be a Boolean
 $\mDatalog(\taubSigma)$-query whose program is in TMNF.
 Within time 1-fold exponential in the size of $Q$ and $\Sigma$, we
 can construct an NBTA $\Aut{A}^{\Yes}$, which accepts exactly those
 ordered $\Sigma$-labeled binary trees $T$ where $\AF{Q}(T)=\Yes$.
\end{Proposition}
\begin{proof}
First, we use Proposition~\ref{query22wata} to construct, within
polynomial time, a 2ATA $\ATAA$ such that for all ordered binary
$\Sigma$-labeled trees $T$, the automaton $\hat{A}$ accepts the
$\hat{\Sigma}$-labeled tree $\hat{T}$ if, and only if,
$\AF{Q}(T)=\Yes$.

Now, we use Theorem~\ref{Thm:2ATAtoNBTA} to construct, within time
1-fold exponential in $\size{\ATAA}$ (i.e., 1-fold exponential in the
size of $Q$ and $\Sigma$), an NBTA $\Aut{A}$ with
$\lang(\Aut{A})=\lang(\ATAA)$. 

Note that $\Aut{A}$ operates on
$\hat{\Sigma}$-labeled trees, while we are looking for an NBTA $\Aut{A}^{\Yes}$
operating on $\Sigma$-labeled trees.
To obtain such an automaton, we proceed as follows:

Let $\Aut{B}$ be an NBTA of alphabet
$\hat{\Sigma}$ which accepts exactly those $\hat{\Sigma}$-labeled
trees $T'$ for which there exists a $\Sigma$-labeled tree $T$ such
that $T'=\hat{T}$ (building such an NBTA is straightforward: the
automaton just needs to check that the $\hat{\Sigma}$-labels correctly
identify the root node, the nodes that are left (right) children, and
the nodes that have no left (right) child).

Using the intersection-construction mentioned in
Fact~\ref{fact:NBTA-constructions}, we can build the intersection automaton
$\Aut{A}'$ of $\Aut{B}$ and $\Aut{A}$. I.e., $\Aut{A}'$ accepts a
$\hat{\Sigma}$-labeled tree $T'$ iff there exists a $\Sigma$-labeled
tree $T$ such that $T'=\hat{T}$, and $\hat{T}$ is accepted by
$\Aut{A}$.

Finally, we use the projection-construction described in
Fact~\ref{fact:NBTA-constructions} to obtain an NBTA $\Aut{A}^{\Yes}$ over
alphabet $\Sigma$, such that $\lang(\Aut{A}^{\Yes})=\setc{\Proj_{\Sigma}(T')}{T'\in\lang(\Aut{A}')}$.
Thus, $\Aut{A}^{\Yes}$ accepts a tree $T$ $\iff$ $\hat{T}$ is accepted by
$\Aut{A}$ $\iff$ $\AF{Q}(T)=\Yes$. 

Since the intersection- and projection-constructions
can be performed within time polynomial in the size of its input
NBTAs, the entire construction of $\Aut{A}^{\Yes}$ takes time at most
1-fold exponential in the size of $Q$ and $\Sigma$.
\end{proof}

\subsection{Step~\eqref{item:Step3}: Finishing the proof of Theorem~\ref{Thm:orderedInExptime}}

\begin{proofof}{Theorem~\ref{Thm:orderedInExptime}}
Our goal is to show that the QCP for unary
$\mDatalog(\tauGK^{\Child})$-queries on ordered trees belongs to
$\EXPTIME$.

Let $\Sigma$, $Q_1$, $Q_2$ be an input for the QCP.
Let $\Sigma'\deff\Sigma\times\set{0,1}$.
By using Proposition~\ref{prop:BinaryTreesSufficeForQCP-upperBound}
we obtain, within linear time, Boolean $\mDatalog(\taubAlph{\Sigma'})$-queries
$Q'_1$ and $Q'_2$ such that $Q_1\subseteq Q_2$ iff $Q'_1\subseteq
Q'_2$, and the programs of $Q'_1$ and $Q'_2$ are in TMNF.

By using Proposition~\ref{prop:mDatalog2Ayes},
we can construct, within time 1-fold exponential in the size of $Q'_1$
and $\Sigma'$, an NBTA $\Aut{A}_1^{\Yes}$, which accepts exactly those
$\Sigma'$-labeled binary trees $T$ where $\AF{Q'_1}(T)=\Yes$.

By using Proposition~\ref{query2notnbuta},
we can construct, within time 1-fold exponential in the size of $Q'_2$
and $\Sigma'$, an NBTA $\Aut{A}_2^{\No}$, which accepts exactly those
$\Sigma'$-labeled binary trees $T$ where $\AF{Q'_2}(T)=\No$.

Now, we use the intersection-construction mentioned in
Fact~\ref{fact:NBTA-constructions} to build the intersection-automaton
$\Aut{B}$ of $\Aut{A}_1^{\Yes}$ and $\Aut{A}_2^{\No}$.
Clearly, $\Aut{B}$ accepts a $\Sigma'$-labeled binary tree $T$ if, and
only if,
$\AF{Q'_1}(T)=\Yes$ and $\AF{Q'_2}(T)=\No$.

Finally, we use the emptiness-test provided by
Fact~\ref{fact:emptiness} to check whether $\lang(\Aut{B})=\emptyset$.
Clearly, this is the case if, and only if, $Q'_1\subseteq Q'_2$, wich
in turn is true iff $Q_1\subseteq Q_2$.

Since the intersection-construction and the emptiness test take only
time polynomial in the size of the input automata, the entire
algorithm for checking whether $Q_1\subseteq Q_2$ runs in time 1-fold
exponential in the size of $\Sigma$, $Q_1$, and $Q_2$. This completes
the proof of Theorem~\ref{Thm:orderedInExptime}
\end{proofof}

\clearpage
\section{Dealing with the descendant-axis: Proof of Theorem~\ref{Thm:usingDesc}}\label{appendix:GettingRidOfDesc}
The aim of this appendix is to prove the following:

\setcounter{restoreAppTheorem}{\value{theorem}}
\setcounter{theorem}{\value{Counter_Thm:usingDesc}}

\begin{Theorem} \textbf{(restated)}
  The QCP for unary $\mDatalog(\tau_u^{\Root,\Leaf,\Desc})$ on
  unordered trees and for unary $\mDatalog(\tauGK^{\Child,\Desc})$ can
  be solved in 2-fold exponential time.
\end{Theorem}

\setcounter{theorem}{\value{restoreAppTheorem}}

Note that $\tau_u^{\Root,\Leaf,\Desc}\subseteq\tauGK^{\Child,\Desc}$. 
Thus, to prove Theorem~\ref{Thm:usingDesc}, it suffices to provide a
2-fold exponential algorithm for the QCP for unary
\linebreak[4]
$\mDatalog(\tauGK^{\Child,\Desc})$-queries on ordered trees.

Upon input of two $\mDatalog(\tauGK^{\Child,\Desc})$-queries $Q_1$ and
$Q_2$, our algorithm proceeds as follows:
First, we transform $Q_1$ and $Q_2$ into equivalent queries $Q'_1$ and
$Q'_2$ that do \emph{not} contain the $\Desc$-predicate.
Afterwards, we use the algorithm provided by
Theorem~\ref{Thm:orderedInExptime} to decide whether $Q'_1\subseteq
Q'_2$.
Thus, Theorem~\ref{Thm:usingDesc} is an immediate consequence of  
Theorem~\ref{Thm:orderedInExptime} and the following Lemma~\ref{Lemma:GettingRidOfDesc}.

\begin{Lemma}\label{Lemma:GettingRidOfDesc}
For every $\mDatalog(\tauGK^{\Child,\Desc})$-query $Q$ there is an
equivalent \linebreak[4] $\mDatalog(\tauGK^{\Child})$-query $Q'$, which can be
computed in 1-fold exponential time.
\end{Lemma}
The remainder of Appendix~\ref{appendix:GettingRidOfDesc} is devoted
to the proof of Lemma~\ref{Lemma:GettingRidOfDesc}.
\\
The proof proceeds in three steps:
\begin{description}
 \item[Step 1:]
   Gottlob, Koch, and Schulz \cite[Theorem~6.6]{GottlobKochSchulz} showed that
   every conjunctive query using the axes $\Child$, $\Desc$,
   $\Ns$ can be rewritten, in 1-fold 
   exponential time, into an equivalent union of \emph{acyclic}
   conjunctive queries.
   We extend their result to monadic datalog rules that may also contain the
   $\Fc$-relation (see Lemma~\ref{Lemma:GKS} below).
 \item[Step 2:]
   Afterwards, we use a result of \cite{GottlobKoch} which shows that
   every \emph{acyclic} conjunctive query can be rewritten, in linear
   time, into a monadic datalog program that is ``almost'' in TMNF
   (see Lemma~\ref{remark_2} below).
 \item[Step 3:]
   Finally, observe that each TMNF-rule which uses the $\Desc$-relation can be
   replaced (in constant time) by two suitable rules using $\Child$ (see
   Fact~\ref{Fact:TMNF-Desc2Child}).
\end{description}

Step~3 is established by the following obvious fact:

\begin{Fact}\label{Fact:TMNF-Desc2Child} \ \\
Over trees, 
the rule
\ $
   X(x) \leftarrow \Desc(x,y),Y(y)
$ \ 
is equivalent to the rules
\begin{align*}
   X(x) &\leftarrow \Child(x,y),Y(y)\\  
   X(x) &\leftarrow \Child(x,y),X(y) . 
\end{align*} 
Similarly, 
the rule
\ $
   X(x) \leftarrow \Desc(y,x),Y(y)
$ \ 
is equivalent to the rules
\begin{align*}
   X(x) &\leftarrow \Child(y,x),Y(y)  \\  
   X(x) &\leftarrow \Child(y,x),X(y) . 
\end{align*} 
\end{Fact}

For Steps~1 and 2, let us recall the notion of
\emph{acyclic} queries considered in
\cite{GottlobKoch,GottlobKochSchulz}.
Let $\tau$ be a schema consisting of relations of arity at most 2.
Let $r$ be a rule of a monadic datalog query of schema $\tau$.
The \emph{directed rule graph} $G_r$ is the multigraph whose vertex
set is the set of variables of $r$, and where for each binary atom of
the form $R(x,y)$ occurring in the rule's body, there is a directed
edge $e_R$ from node $x$ to node $y$. 
The \emph{shadow} of $G_r$ is the undirected multigraph obtained from
$G_r$ by ignoring the edge directions.
We say that $r$ contains 
a \emph{directed cycle} if the multigraph $G_r$ contains a directed
cycle.
Accordingly, $r$ contains an \emph{undirected cycle} if the shadow of
$G_r$ contains a cycle. A rule is called \emph{acyclic} if it does not
contain an undirected cycle;
an $\mDatalog(\tau)$-program is \emph{acyclic} if all its rules are acyclic.

Step~2 of our agenda is provided by the following lemma.

\begin{Lemma}[{\cite[Lemma~5.8]{GottlobKoch}}]\label{remark_2}
 Let $r$ be an acyclic 
 monadic datalog rule over relations that are
 either unary or binary. Then, $r$ can be decomposed in linear time
 into a monadic datalog program in which each rule is one of the
 three forms 
 \[
   X(x) \leftarrow R(y,x),Y(y) 
   \qquad
   X(x) \leftarrow  R(x,y),Y(y) 
   \qquad
   X(x)\leftarrow Y(x), Z(z)
 \]  
   where $x$ (resp., $Y$) may but does not have to be different from
   $z$ (resp., $Z$).
\end{Lemma}

Finally, Step~1 of our agenda is established by the following
Lemma~\ref{Lemma:GKS}, which generalises a result by
Gottlob, Koch, and Schulz \cite[Theorem~6.6]{GottlobKochSchulz} to 
queries that may make use of the $\Fc$-predicate.

\begin{Lemma}\label{Lemma:GKS}
 Every unary $\mDatalog(\tauGK^{\Child,\Desc})$-query $Q$ can be
 rewritten, in 1-fold exponential time, into an equivalent
 $\mDatalog(\tauGK^{\Child,\Desc})$-query $Q'$ such that each rule in the program
 of $Q'$ is acyclic.
\end{Lemma}
\begin{proof}
Let $\PP$ be the program of $Q$. We choose $Q'$ to have the same
query predicate as $Q$. The program $\PP'$ of $Q'$ is constructed as follows.

\medskip
\noindent

We initialise $\PP'$ to be equal to $\PP$. 
Then, while $\PP'$ is \emph{not} acyclic, do the following:
Let $r$ be a rule in $\PP'$ that is not acyclic. Remove $r$ from $\PP'$.

\medskip

\noindent
\emph{Case 1:} \ If $r$ contains a directed cycle, note that $r$ is
not~satisfiable (since the directed cycle is built from the axes
$\Fc$, $\Ns$, $\Child$, $\Desc$). 
\\
Thus, we simply drop $r$.

\medskip

\noindent
\emph{Case 2:} \ Otherwise, $r$ must contain an undirected cycle, but 
no directed cycle. 
Then, the directed query graph $G_r$ is a DAG, and 
there must exist a variable $z$ of $r$ which belongs to an undirected cycle,
such that $G_r$ contains no directed path
from $z$ to another variable that belongs to an undirected cycle. 
For this variable $z$, the rule's body must contain two atoms of the form
$R(x,z)$ and $S(y,z)$ (where $R,S\in\set{\Fc,\Ns,\Child,\Desc}$, and
$x,y$ are variables). We make the following case distinction:
\begin{enumerate}[(i)]
 \item
  In case that $R=\Fc$ and $S=\Ns$ (or vice versa), note that the rule
  is unsatisfiable, and hence we simply drop $r$.
  \smallskip
 \item\label{item:ii}
  In case that $R=S\in\set{\Fc,\Ns,\Child}$, note that $R(x,z)\und
  S(y,z)$ is equivalent to $R(x,z)\und y{=}x$. Thus, we let
  $\tilde{r}$ be the rule obtained from $r$ by omitting the atom
  $S(y,z)$ and replacing all occurrences of $y$ by $x$. We add
  $\tilde{r}$ to $\PP'$.
  \smallskip
 \item
  In case that $R=S=\Desc$, note that $R(x,z)\und S(y,z)$ is
  equivalent to $\varphi\deff$
  \[
     \big(\Desc(x,y)\und\Desc(y,z)\big)
     \oder
     \big(\Desc(y,x)\und\Desc(x,z)\big)
     \oder
     \big(\Desc(x,z)\und y{=}z\big).
  \]
  For each $i\in\set{1,2,3}$ we let $\tilde{r}_i$ be the rule obtained
  from $r$ by replacing ``$R(x,z),S(y,z)$'' with the $i$-th clause of
  $\varphi$. Concerning $\tilde{r}_3$, we furthermore delete the atom $y{=}z$ and replace all occurrences of $y$
  by $z$. We add $\tilde{r}_1$, $\tilde{r_2}$, and $\tilde{r}_3$ to $\PP'$.
  \smallskip
 \item
  In case that $R=\Fc$ and $S=\Child$ (or vice versa), note that 
  $R(x,z)\und S(y,z)$ is equivalent to $R(x,z)\und y{=}z$. Hence, we
  proceed in the same way as in case~\eqref{item:ii}.
  \smallskip
 \item
  In case that $R=\Ns$ and $S\in\set{\Child,\Desc}$ (or vice versa), note that
  $R(x,z)\und S(y,z)$ is equivalent to $R(x,z)\und S(y,x)$.
  We let $\tilde{r}$ be the rule obtained from $r$ by replacing the
  atom $S(y,z)$ with the atom $S(y,x)$, and we add $\tilde{r}$ to
  $\PP'$.
  \smallskip
 \item
  In case that $R\in\set{\Fc,\Child}$ and $S=\Desc$ (or vice versa),
  note that $R(x,z)\und S(y,z)$ is equivalent to $\varphi\deff$
  \[
    \big( R(x,z)\und \Desc(y,x) \big)
    \oder
    \big( R(x,z)\und y{=}z \big).
  \]
  For each $i\in\set{1,2}$ we let $\tilde{r}_i$ be the rule obtained
  from $r$ by replacing ``$R(x,z),S(y,z)$'' with the $i$-th clause of
  $\varphi$.
  Concerning $\tilde{r}_2$, we furthermore delete the atom $y{=}z$ and replace all occurrences of
  $y$ by $z$. We add $\tilde{r}_1$ and $\tilde{r}_2$ to $\PP'$.
\end{enumerate}

\noindent
Clearly, the obtained query $Q'$ is equivalent
to the original query $Q$. Furthermore, along the same lines as in the
proof of \cite[Lemma~6.4]{GottlobKochSchulz}, one can show that the
algorithm terminates after a number of steps that is at most 1-fold
exponential in the size of the input query $Q$. Of course, upon
termination the program $\PP'$ is acyclic.
Thus, the proof of Lemma~\ref{Lemma:GKS} is complete.
\end{proof}

Finally, we are ready for the proof of Lemma~\ref{Lemma:GettingRidOfDesc}.

\begin{proofof}{Lemma~\ref{Lemma:GettingRidOfDesc}}
Let $Q$ be the given $\mDatalog(\tauGK^{\Child,\Desc})$-query.
\\
Using Lemma~\ref{Lemma:GKS} we construct, within 1-fold exponential time,
an equivalent $\mDatalog(\tauGK^{\Child,\Desc})$-query $Q_1$
such that each rule in the program $\PP_1$ of $Q_1$ is acyclic.
By applying Lemma~\ref{remark_2} to each rule of $\PP_1$, we obtain an
equivalent $\mDatalog(\tauGK^{\Child,\Desc})$-query $Q_2$ such that 
each rule in the program $\PP_2$ of $Q_2$ is one of the following forms:
 \[
   X(x) \leftarrow R(y,x),Y(y) 
   \qquad
   X(x) \leftarrow R(x,y),Y(y)
   \qquad
   X(x)\leftarrow Y(x), Z(z)
 \]  
with $R\in\set{\Fc,\Ns,\Child,\Desc}$.
\\
Applying Fact~\ref{Fact:TMNF-Desc2Child}, we then replace every
rule of $\PP_2$ that contains the $\Desc$-relation by two rules that
use the $\Child$-relation. 
\\
This leads to an
$\mDatalog(\tauGK^{\Child})$-query $Q_3$ that is equivalent to
$Q$. Furthermore, $Q_3$ is computed in time 1-fold exponential in the
size of $Q$. 
\end{proofof}